\documentclass[10pt,journal]{IEEEtran}

\usepackage{graphics,
           psfrag,
           epsfig,
           amsthm,
           cite,
           amssymb,
           url,
           dsfont,
           subfigure,
           subfig,
           algorithm,
           algorithmic,
           balance,
           enumerate,
           color,
           setspace
}
\usepackage{amsmath}

\newtheorem{theorem}{Theorem}

\newtheorem{property}{Property}

\newcommand{\sref}[1]{Section~\ref{#1}}
\newcommand{\appref}[1]{Appendix~\ref{#1}}
\newcommand{\fref}[1]{Figure~\ref{#1}}

\newcommand{\prref}[1]{Property~\ref{#1}}
\newcommand{\cref}[1]{Constraint~\ref{#1}}
\newcommand{\thref}[1]{Theorem~\ref{#1}}
\newcommand{\lref}[1]{Lemma~\ref{#1}}

\newcommand{\algref}[1]{Algorithm~\ref{#1}}

\newcommand{\ignore}[1]{}

\IEEEoverridecommandlockouts
\begin{document}

\title{Completion Delay Minimization \\ for Instantly Decodable Network Codes}
\author{Sameh~Sorour,~\IEEEmembership{Student Member,~IEEE,}
        Shahrokh~Valaee,~\IEEEmembership{Senior Member,~IEEE}
\ignore{\IEEEcompsocitemizethanks{\IEEEcompsocthanksitem The authors are with Edward S. Rogers Sr. Department of Electrical and Computer Engineering,
    University of Toronto, 10 King's College Road, Toronto, ON, M5S 3G4, Canada,
    e-mail:\{samehsorour, valaee\}@comm.utoronto.ca.
    \IEEEcompsocthanksitem This work is an extension to our paper \cite{ICC10} in ICC 2010}}
    \thanks{The authors are with Edward S. Rogers Sr. Department of Electrical and Computer Engineering,
    University of Toronto, 10 King's College Road, Toronto, ON, M5S 3G4, Canada,
    e-mail:\{samehsorour, valaee\}@comm.utoronto.ca.}
    \thanks{This work is an extension to our paper \cite{ICC10} in ICC 2010}
     }

\maketitle

\IEEEoverridecommandlockouts

\begin{abstract}
In this paper, we consider the problem of minimizing the completion delay for \emph{instantly decodable network
coding (IDNC)}, in wireless multicast and broadcast scenarios. We are interested in this class of network coding due to its numerous benefits, such as low decoding delay, low coding and decoding complexities and simple receiver requirements. We first extend the IDNC graph, which represents all feasible IDNC coding opportunities, to efficiently operate in both multicast and broadcast scenarios. We then formulate the minimum completion delay problem for IDNC as a stochastic shortest path (SSP) problem. Although finding the optimal policy using SSP is intractable, we use this formulation to draw the theoretical guidelines for the policies that can efficiently reduce the completion delay in IDNC. Based on these guidelines, we design a maximum weight clique selection algorithm, which can efficiently reduce the IDNC completion delay in polynomial time. We also design a quadratic time heuristic clique selection algorithm, which can operate in real-time applications. Simulation results show that our proposed algorithms efficiently reduce the IDNC completion delay compared to the random and maximum-rate algorithms, and almost achieve the global optimal completion delay performance over all network codes in broadcast scenarios.
\end{abstract}

\newcounter{mytempeqncnt}

\section{Introduction and Motivation} \label{sec:intro}
\ignore{
\IEEEPARstart{T}{he} design of network coding (NC) algorithms that optimize different system parameters over wireless erasure channels, has recently been an intensive area of research
\cite{4313060,\ignore{4595303},4544612,4895447,4476183,Drinea2009,Sadeghi2009,Sadeghi2010,4397057,5152148,Sundararajan2009}.
In these works, network coding was opportunistically optimized, with respect to the side information available at the different receivers, to improve different performance metrics such as decoding delay, completion delay, video streaming quality and sender backlog. However, the proposed solutions were either for error-free channels (index coding) \cite{4313060,\ignore{4595303},4544612}, dependable on prior knowledge of future receiver errors (offline algorithms) \cite{4476183,Drinea2009}, or optimizing the considered parameter for only one step \cite{Sadeghi2009,Sadeghi2010}. The few works that looked into optimizing parameters on lossy long range transmission scenarios \cite{4397057,5152148} did not propose efficient solutions that can be applied in real-time for typical network settings. Moveover, all previous
works (except for those on index coding) did not consider multicast scenarios, in which receivers are interested in receiving only a subset of the packets rather than all the packets as in the broadcast scenario.
}

\IEEEPARstart{M}{ulticast} Broadcast Services (MBS) have become a corner stone in the design of all future wireless and mobile standards and networks, such as LTE and WiMAX. It is now very common to find a group of pedestrians or people on a bus watching a soccer match on their smart phones while some of them are downloading files, others are watching the breaking news and others are looking for the closest restaurants or using voice over IP applications,... etc. Due to the high demand on these MBS applications and their high bandwidth and delay requirements, it is very important for MBS protocols to not only efficiently utilize the scarce bandwidth resources available to the network but to allow progressive packet reception to satisfy delay requirements. In other words, while these protocols are minimizing the amount of resources (e.g. number of transmissions) consumed by such application to increase the bandwidth efficiency, they should also be able to guarantee the quality of streaming and delay-intolerant applications, in which the received packets should be always useful at their reception instant, in order to prevent interruption or flickering of the stream.\ignore{ In other words, the streaming packets should not experience high delays until their correct reception and passing to upper layers.} The simultaneous achievement of these goals calls for new approachs to increase the efficiency of the packet transmission and recovery processes. One major breakthrough in this area came with the development of network coding (NC).

Recently, NC has shown great abilities to substantially improve transmission efficiency, packet recovery, throughput and delay over broadcast erasure channels \cite{4313060,4595303,4544612,4895447,4476183,Drinea2009,Sadeghi2009,Sadeghi2010,4397057,5152148,Sundararajan2009,ISIT09,TON09,Wang2010,Gatzianas2010}.
Two trends of network coding can be distinguished in the literature, namely random (or full) network coding (RNC) \cite{1228459}\cite{4015713} and opportunistic network coding (ONC)\cite{1159942}\cite{Katti2005}. The former trend combines all the packets with random non-zero and independent coefficients in each transmission. The latter trend exploits the receivers' side information in selecting packets to be coded in each transmission to achieve a certain target. Despite the great interest in RNC in the literature, its ability to recover packets without feedback, and its optimality in reducing the number of packet transmissions in broadcast scenarios \cite{4675714}, it is only feasible for applications with high delay tolerance, since it does not support progressive packet decoding\ignore{, required in streaming applications}. It is also inefficient\ignore{ also fails in minimizing the required resources} in unicast and multicast scenarios, in which different groups of receivers are interested in different subsets of the transmitted packets \cite{TON09}.

In this paper, we are interested in a sub-class of opportunistic network coding, called the \emph{instantly decodable network coding (IDNC)}, in which received packets are allowed to be decoded only at their reception instant and cannot be stored for future decoding. This IDNC approach is currently attracting much attention \cite{5425315,Sadeghi2009,Sadeghi2010} due to its numerous desirable properties. First, IDNC provides instant packet recovery upon appropriate packet reception, a property that perfectly matches the requirements of the MBS streaming applications, and that RNC and general ONC lack.\ignore{ It is also attractive in several other applications such as roadside to vehicle safety message broadcast and coordinated command dissemination to sensors, in which the reduction of the packet decoding delay is important.} Moreover, the IDNC encoding can be implemented using binary XOR, which eliminates the complicated operations over large Galois fields and the coefficient reporting overhead\ignore{, required by linear NC}. This XOR encoding also simplifies the decoding process at the receivers,\ignore{ as each receiver can simply cancel out the packets it already knows. This} and eliminates the need for matrix inversion\ignore{ at the receivers, which is}, which represents a computational bottleneck in RNC \cite{Sadeghi2009}. Finally, no buffers are needed at the receivers to store coded packets for future decoding. These simple decoding and bufferless properties allow the design of simple and cost-efficient receivers, which is an important requirement for mobile hand-held devices.

Despite its attractive properties, IDNC does not guarantee service to all receivers in each transmission, which affects its ability to both maximize the number of decoding receivers in each and every transmission (i.e. capacity) and minimize the number of transmissions to deliver a frame or packets (i.e. completion delay). The former problem was deeply investigated in IDNC \cite{5425315,Sadeghi2009,Sadeghi2010} and ONC \cite{Wang2010,Gatzianas2010}, whereas limited work has addressed the latter problem only for erasure-less channels (a.k.a. the index coding problem \cite{4313060,4595303,4544612}). It has been shown that this index coding problem is NP-Hard to solve and to approximate \cite{4031356,4594999}. Naturally, this complexity becomes worse in case of erasure channels, which leaves us no choice other than designing efficient heuristics to solve it.

This fact raises the following question that we address in this paper: \emph{What is the efficient heuristic policy that can reduce the expected completion delay in IDNC, over erasure channels, for both multicast and broadcast scenarios?} Intuitively, one might think that the best heuristic policy is to maximize the number of receivers that can decode a new packet in each transmission, as studied in \cite{5425315,Sadeghi2009,Sadeghi2010,Wang2010,Gatzianas2010}. In this paper, we show that this intuition is not true and that the solution to the completion delay problem is obtained by giving priority to targeting the receivers with higher demands and worse channels. To reach this result, we first extend the IDNC graph, which represents all the feasible IDNC packet combinations according to receivers' side information, to suit both multicast and broadcast scenarios. We then formulate the minimum completion delay problem in IDNC as a \emph{stochastic shortest path (SSP) problem}, which is a special case of the Markov decision process (MDP), having absorbing states. Although this formulation is impossible to solve, we mainly employ it to draw the theoretical properties of the policy that can efficiently reduce the completion delay in IDNC.

Based on these properties, we design a two-stage maximum weight clique selection algorithm to reduce the completion delay in IDNC in polynomial time for moderate graph sizes. For further complexity reduction, we design a quadratic time heuristic algorithm, based on greedy maximum weight vertex search, which is more suitable for real-time applications. We finally compare the performance of our proposed optimal and heuristic maximum weight clique search algorithms to RNC, the random IDNC algorithm (that selects served receivers randomly) and the maximum clique IDNC algorithm (maximizing the number of decoding receivers).

The contributions of this papers are summarized as follows:
\begin{itemize}
\item To the best of our knowledge, this paper presents the first rigorous study on reducing the IDNC completion delay over erasure channels. In this study, we do not limit ourselves to maximizing the number of decoding receivers\ignore{in each transmission} as in \cite{5425315,Sadeghi2009,Sadeghi2010,Wang2010,Gatzianas2010} but rather investigate both the order of receiver service and the evolution of coding opportunities along the transmission process, which were shown to be the key factors affecting the optimization of completion delay in IDNC.
\ignore{\item We show that targeting the maximum number of receivers in each transmissions is not the good solution for the problem but rather targeting the receivers with larger sets of missing packets and larger erasure probabilities.}
\item We design\ignore{ simple} polynomial and quadratic-time heuristics that\ignore{ are shown to} achieve near optimal completion delay performance.
\end{itemize}

The rest of the paper is organized as follows. We first summarize related works in
\sref{sec:related}. In \sref{sec:model}, we introduce the system model and parameters. The IDNC graph is illustrated in \sref{sec:GIDNC-graph}. We present the problem formulation in \sref{sec:SSP-formulation} and draw the properties of efficient IDNC completion delay reduction in Sections \ref{sec:properties}, \ref{sec:geometry} and \ref{sec:graph-evolution}. The proposed algorithms are introduced in \sref{sec:algorithms} and their performances are evaluated in \sref{sec:simulations}. Finally, \sref{sec:conclusion} concludes the paper.

\section{Related Work} \label{sec:related}
Since its first introduction in \cite{850663}, network coding has been a great attraction to
numerous studies as a routing and scheduling scheme that attains maximum information flow in a
network.\ignore{ The core of network coding is the idea of packet mixing using several techniques such as packet XOR \cite{1159942} and linear coding \cite{1176612}.}
In \cite{4313060,4595303,4544612}, the problem of determining packet combinations, to minimize the number of transmissions (i.e. completion delay) over erasure-less channels, was studied under the name of ``index coding''. In \cite{4031356}, it has been shown that finding the optimal solution of the index coding problem is NP-hard and thus different heuristics to solve the index coding problem were proposed in \cite{4544612}. In this paper, we extend the study to the case of erasure channels. Our problem differs from index coding in that the feedback status of different receivers changes probabilistically after each transmission over erasure channels. Thus, the coded packets cannot be scheduled for the whole transmission process all at once, as in index coding, but rather require to be dynamically scheduled after each transmission according to the received feedback.\ignore{ In \cite{ISIT09,TON09}, we proposed adaptive selection algorithms between random IDNC and RNC to reduce the completion delay. However, the problem of minimizing the completion delay in IDNC is still not solved. In \cite{4549741}, the bandwidth efficiencies of different network coded transmission schemes were derived.}

In \cite{4895447}, the authors proposed an online network coding algorithm for the three-receiver case, proved its rate optimality and conjectured\ignore{ that it achieves an} its asymptotically optimal average delay. In \cite{4476183} and \cite{Drinea2009}, the decoding delay performance of offline algorithms was analyzed and the decoding delay of several greedy online NC algorithms were compared for i.i.d. erasure channels. These proposed algorithms performed un-prioritized packet selection for each NC transmission and did not consider the channel conditions in their selection procedures. \cite{Sadeghi2009,Sadeghi2010} proposed a prioritized and channel-aware packet selection algorithm that achieves optimal decoding delay for a more strict version of IDNC. All these works are clearly different from our problem in terms of objective and proposed solutions.

For a more general ONC scenario than IDNC, in which un-decoded packets can be stored for future use, \cite{Wang2010,Gatzianas2010} study the maximization of service rates (i.e. capacity) of multiple unicast sessions over 1-to-$K$ broadcast erasure channels. Inner and outer capacity bounds were derived and were shown to meet in the special cases of symmetric and spatially independent erasure channels. \cite{Wang2010,Gatzianas2010} also proposed capacity achieving packet evolution algorithms, serving subsets of receivers with incremental sizes in order to maximize the number of decoding receivers in every transmission, and thus their rates. Despite the more general decodability assumption considered in these works, our paper differs from them in both the optimization objective and proposed algorithms. Unlike the aim of \cite{Wang2010,Gatzianas2010} to maximize the achievable unicast rates, through maximizing the number of decoding receivers in every transmission, our paper studies the problem of minimizing the total number of transmissions (i.e. completion delay) to deliver a frame of multicast or broadcast packets over erasure channels, which makes it an extension to the index coding problem. Moreover, the proposed packet evolution algorithms in \cite{Wang2010,Gatzianas2010} do not prioritize receiver service, but rather serve receiver subsets, with same size, in an \emph{arbitrary} sequential cyclic or acyclic fashion. Consequently, these algorithms are not suitable solvers to the completion delay minimization problem, whose solution greatly depends on the prioritization of receiver service according to their demand and erasure probabilities, as will be shown in \sref{sec:geometry}. Our proposed algorithms focus on implementing this prioritization and are thus significantly different from those proposed in \cite{Wang2010,Gatzianas2010}.

\ignore{
our parameter of interest greatly depends not only on per transmission benefits to maximize capacity but mainly on the order of receiver service and the evolution of coding opportunities along the frame transmission, which are naturally not investigated for capacity studies in \cite{Wang2010,Gatzianas2010}. Moreover, we focus on a specific set of ONC, namely IDNC, for its strict aim of instant packet recovery required in many practical applications.}

\ignore{In all previous studied online algorithms, the packet selection is optimized to increase the benefit per transmission. On the other hand, In \cite{4397057,5152148}, a Markov decision processes (MDP) \cite{Puterman1994} model was employed to find the optimal NC selection policy that minimizes the distortion of video streams over a finite transmission horizon. However, the dimensionality of the MDP's state and action spaces makes the computation of these optimal policies intractable, both online and offline, for typical network settings.\ignore{ In \cite{5152148}, a simulation based dynamic programming algorithm was proposed to reduce the computational complexity. However, the resulting complexity of the proposed algorithm is still intractable.}}

\section{System Model and Parameters} \label{sec:model}
The system model we consider in this paper\ignore{ is depicted in \fref{fig:model}.
\begin{figure}[t]
\centering
  \includegraphics[width=0.8\linewidth]{model}\\
  \caption{System model and packet status after the initial transmission phase}\label{fig:model}
\end{figure}
It} consists of a wireless sender that is required to deliver a frame (denoted by
$\mathcal{N}$) of $N$ source packets to a set (denoted by $\mathcal{M}$) of $M$ receivers.
Each receiver is interested in receiving either a subset or all the packets of $\mathcal{N}$. The
former case is referred to as ``multicast'' whereas the latter case is referred to as ``broadcast''. We will refer to the requested and undesired packets of any receiver by its ``primary'' and ``secondary'' packets.
The sender initially transmits the $N$ packets of the frame uncoded in an \emph{initial
transmission phase}. Each receiver listens to all transmitted packets (even the ones that it does not want) and feedbacks to the sender a positive acknowledgement (ACK) for each received packet. At the end of the initial transmission phase, three sets of packets are attributed to each receiver $i$:
\begin{itemize}
\item The \emph{Has} set ($\mathcal{H}_i$) is defined as the set of primary and secondary packets correctly received by receiver $i$.\ignore{ This set includes both desired and undesired packets by this receiver.}
\item The \emph{Lacks} set ($\mathcal{L}_i = \mathcal{N} \setminus \mathcal{H}_i$) is defined as the set of primary and secondary packets\ignore{ that are} not received by $i$.\ignore{ In other words, $\mathcal{L}_i = \mathcal{N} \setminus \mathcal{H}_i$.}
\item The \emph{Wants} set ($\mathcal{W}_i\subseteq\mathcal{L}_i$) is defined as the set of primary packets that receiver $i$\ignore{ wants to receive but} has not yet received.\ignore{ Note that $\mathcal{W}_i\subseteq\mathcal{L}_i$.}
\end{itemize}
The sender stores this information in a \emph{state feedback matrix (SFM)} $\mathbf{F} =
\left[f_{ij}\right],~\forall~i\in\mathcal{M},j\in\mathcal{N}$ such that $f_{ij} = 0$ if $j\in\mathcal{H}_i$, $f_{ij} = 1$ if $j\in\mathcal{W}_i$, and $f_{ij} = -1$ if $j \in \mathcal{L}_i\setminus \mathcal{W}_i$.
\ignore{
\begin{equation} \label{eq:SFM}
f_{ij} =
\begin{cases}
0 \qquad &j \in \mathcal{H}_i\\
1 \qquad &j \in \mathcal{W}_i\\
-1 \qquad &\mbox{otherwise}\;. \ignore{j \in \mathcal{L}_i(s)\setminus \mathcal{W}_i(s)}
\end{cases}
\end{equation}}

After the initial transmission phase, a recovery transmission phase starts, in which the
sender exploits the reception diversity in the SFM to employ NC. These NC
packets must include at most one source packet from the Wants or Lacks sets of a subset or all of the
receivers. The receivers that cannot decode a new source packet from this NC packet discard it.
For each decoded source packet, the receivers send ACK packets that are used by the sender to update the SFM and the\ignore{ attributed} sets $\mathcal{H}_i$, $\mathcal{L}_i$ and $\mathcal{W}_i$, $\forall~i$.
This process is repeated until all receivers obtain their requested packets. We define the \emph{completion delay} of a frame as the number of recovery transmissions required to deliver all requested packets to their receivers.

Define $\boldsymbol{\varrho} = \left[\varrho_1,\dots,\varrho_M\right]$, $\boldsymbol{\varphi} = \left[\varphi_1,\dots,\varphi_M\right]$ and $\boldsymbol{\psi} = \left[\psi_1,\dots,\psi_M\right]$ as the Has, Lacks and Wants vectors, such that $\varrho_i$, $\varphi_i$ and $\psi_i$ are the cardinalities of $\mathcal{H}_i$, $\mathcal{L}_i$ and $\mathcal{W}_i$, respectively. Let $p_{i}$ and $q_i = 1-p_i$ be the packet erasure and success probabilities observed by receiver $i$, respectively. We assume that\ignore{ the erasure probabilities} $p_i$ and $q_i~\forall~i$ do not change during the frame transmission period. Also, let $\mu_i$ be the demand ratio of receiver $i$, defined as the ratio of its primary packets in the frame to the frame size $N$. Given this definition of $\mu_i$, we can focus on studying the multicast scenario, since the broadcast scenario can be viewed as a special case of the multicast scenario, in which\ignore{ the demand ratios for all receivers is equal to} $\mu_i = 1~\forall~i$.\ignore{ Consequently, we will study the multicast scenario in our formulations and algorithm designs, then will test them for the broadcast scenario in \sref{sec:simulations} by setting all the demand ratios to 1.} Finally, define $\mu = \frac{1}{M}\sum_{i=1}^M \mu_i$ as the average of the demand ratios of all receivers.

\section{IDNC graph} \label{sec:GIDNC-graph}
To form optimized IDNC packets, we should first design a representation of all feasible packet combinations that are instantly decodable by any subset or all the receivers\ignore{, in the multicast scenario}. An initial idea about the representation of packet combinations was introduced in the form of a graph, when designing a heuristic algorithm to solve the index coding problem \cite{4544612,4313060}. This graph, which we will denote by $\mathcal{G}_\rho(\mathcal{V}_\rho,\mathcal{E}_\rho)$, is constructed by first inducing a vertex $v_{ij}$ in $\mathcal{V}_{\rho}$ for each packet $j \in \mathcal{W}_i$, $\forall~i\in\mathcal{M}$. Two vertices $v_{ij}$ and $v_{kl}$ in $\mathcal{G}_\rho$ are connected by an edge in $\mathcal{E}_\rho$ if one of the following conditions is true:
\begin{itemize}
\item C1: $j = l$ $\Rightarrow$ The two vertices are induced by the loss of the same packet $j$ by two different receivers $i$ and $k$.
\item C2: $j\in \mathcal{H}_k$ and $l \in \mathcal{H}_i$ $\Rightarrow$ The requested packet of each vertex is in the Has set of the receiver\ignore{ that induced} of the other vertex.
\end{itemize}
Consequently, each edge between two vertices in the graph represents a \emph{coding opportunity}, which is defined as an opportunity of generating an instantly decodable packet for the two receivers inducing these vertices. Given this graph, we can easily define the set of all feasible packet combinations in IDNC as the set of packet combinations defined by all maximal cliques in $\mathcal{G}_\rho$ (a maximal clique is a clique that is not a subset of any larger clique).\ignore{ Consequently,} The sender can generate an IDNC packet for a given transmission by XORing all the packets identified by the vertices of a selected maximal clique $\kappa_\rho$ in $\mathcal{G}_\rho$.

The above formulation of $\mathcal{G}_\rho$ is suitable when optimizing packet combinations in a broadcast setting as in \cite{ICC10\ignore{,GC10}}. \ignore{However, the optimization of long term packet combinations in a multicast scenario necessities the consideration of secondary packet delivery to the receivers that are not considered for primary packet reception.}
In multicast scenarios, we can explore the enhancement of coding opportunities at receivers that are not considered for primary packet reception, by delivering secondary packets to them. Although these packets are not requested at these receivers, their reception along the steps of the recovery phase, when they are not targeted with primary packets, will enlarge their Has sets. According to Condition C2, this will increase chances of creating more coding opportunities that can serve these receivers in the future steps towards completion. However, this service of secondary packets should never affect the instant decodability of the primary packets at the other receivers.

To achieve both goals, we propose a new two-layered graph $\mathcal{G}(\mathcal{V},\mathcal{E})$. The primary layer consists of graph $\mathcal{G}_\rho$, described above. The secondary layer $\mathcal{G}_\sigma(\mathcal{V}_\sigma,\mathcal{E}_\sigma)$ is constructed by generating a vertex $v_{ij} \in \mathcal{V}_\sigma$ for each packet $j \in \mathcal{L}\setminus\mathcal{W}_i$, $\forall~i\in\mathcal{M}$, and connecting any two vertices\ignore{ $v_{ij}$ and $v_{kl}$} satisfying either C1 or C2. Finally, we connect any two vertices from both layers if either C1 or C2 holds. In the rest of the paper, we will refer to $\mathcal{G}_\rho$, $\mathcal{G}_\sigma$ and $\mathcal{G}$ as the primary, secondary and IDNC graphs, respectively. \fref{fig:IDNC-graph} depicts an example of a feedback table and its corresponding IDNC graph. It is easy to show that the overall complexity of graph construction is $O(M^2N)$.
\begin{figure}[t]
\centering
  \includegraphics[width=0.65\linewidth]{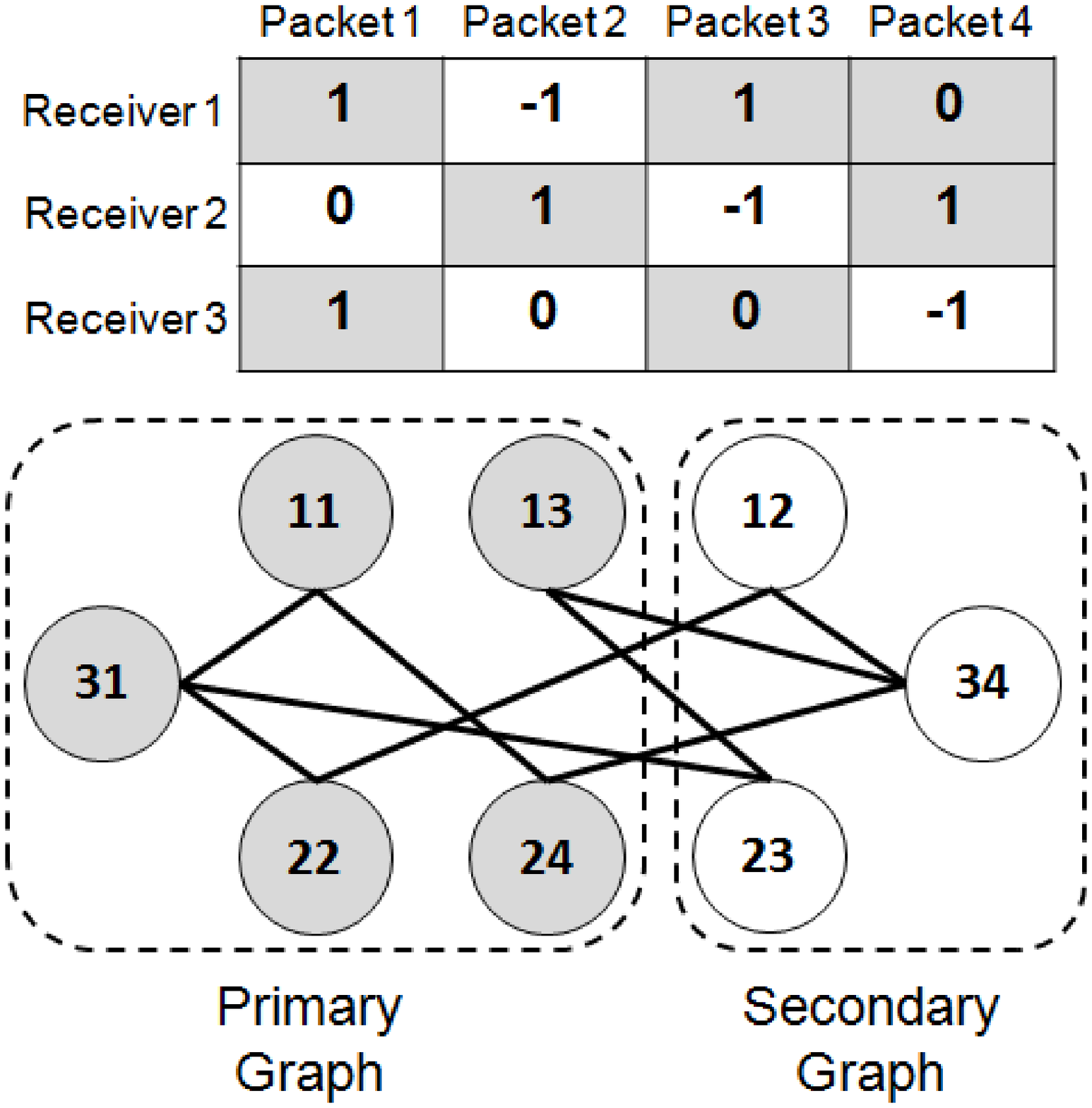}\\
  \caption{Example of a feedback matrix and its corresponding IDNC multicast graph. The shaded and white boxes and vertices represent the wanted and unwanted packets, respectively.}\label{fig:IDNC-graph}
\end{figure}

According to the design of $\mathcal{G}$, we can easily infer that each receiver can have at most one primary or secondary vertex in each of its maximal cliques\ignore{, to maintain instant decodability}. Consequently, the selection of a maximal clique for a given transmission is equivalent to the selection of a set of targeted receivers. We can thus envision the completion delay minimization problem, given the IDNC graph, as a problem of finding the optimal scheduling of targeted receiver subsets under the IDNC constraint, in order to complete the recovery phase as fast as possible. After selecting the maximal clique determining the targeted receivers for any transmission, the source packets identified by its vertices will be XORed to generate the coded packet.

In the rest of the paper, we say that a receiver is \emph{targeted} by an IDNC packet if the selected maximal clique includes a vertex induced by this receiver. We also define $\mathcal{T}_\rho(\kappa)$, $\mathcal{T}_\sigma(\kappa)$ and $\mathcal{T}(\kappa)$ as the set of all primary, secondary and overall targeted receivers of a given maximal clique $\kappa$.\ignore{ In the following sections, we will study the optimal maximal clique (or targeted receiver) selection over graph $\mathcal{G}$ to minimize the expected completion delay.}

\section{Problem Formulation using SSP} \label{sec:SSP-formulation}
\subsection{The SSP Problem}
The stochastic shortest path (SSP) problem is a special case of the infinite horizon MDP, which can
model decision based stochastic dynamic systems with terminating situations. In SSP, all the possible situations the system may encounter are modeled as states in a state space $\mathcal{S}$. In each state $s\in\mathcal{S}$, the system
must select an action $a$ from an action space $\mathcal{A}(s) \subseteq \mathcal{A}$ that will
charge it an immediate cost $c(s,a)$ (where $\mathcal{A}$ denotes the action space of SSP). The
terminating situations of the system can be thus represented as zero-cost absorbing states. Once an action $a$ is taken at state $s$, the system can move to a state $s'$ with probability $P_a(s,s')$, which only depends on the current state and the taken action. An SSP policy $\pi = [\pi(s)]$ is a mapping from $\mathcal{S}\rightarrow\mathcal{A}$ that specifies a given action to each of the states. The optimal policy $\pi^*$ of an SSP is the one that minimizes
the cumulative mean cost until an absorbing state is reached.

The algorithms that solve SSPs define a value function $V_{\pi}(s)$ as the expected cumulative cost
until absorption, when the system starts at state $s$ and follows policy $\pi$. It can be
recursively expressed $\forall~s\in\mathcal{S}$ as:
\begin{equation} \label{eq:value-func}
V_{\pi}(s) = c(s,\pi(s)) + \sum_{s'\in\mathcal{S}(s,a)}P_{\pi(s)}(s,s')\:V_{\pi}(s')\;,
\end{equation}
where $\mathcal{S}(s,a)$ is the set of successor states to $s$ when action $a$ is taken
$\left(\mbox{i.e.}~ \mathcal{S}(s,a) = \left\{s'|P_a(s,s')>0\right\}\right)$. Consequently, the
optimal policy at state $s$ can be defined $\forall~s\in\mathcal{S}$ as:
\begin{equation} \label{eq:opt-policy}
\pi^*(s) = \arg\min_{a\in\mathcal{A}(s)} \left\{c(s,a) +
\sum_{s'\in\mathcal{S}(s,a)}P_{a}(s,s')\:V_{\pi^*}(s')\right\}\;.
\end{equation}

\subsection{Problem Formulation} \label{sec:formulation}
The problem of minimizing the expected completion delay in IDNC can be formulated as an SSP problem as follows:\\

\subsubsection{State Space $\mathcal{S}$}
States are defined by all possibilities of SFM $\mathbf{F}(s)$ that may occur during the recovery
transmission phase. For state $s$, the matrix represents the content of Has, Lacks and Wants sets in $s$
(i.e. $\mathcal{H}_i(s)$, $\mathcal{L}_i(s)$ and $\mathcal{W}_i(s)$ $\forall~i\in\mathcal{M}$) as defined by
\eqref{eq:SFM}. According to its definition, the state space has a size of $|\mathcal{S}|=O\left(2^{MN}\right)$. We can characterize each state $s$ by its Has, Lacks and Wants vectors $\boldsymbol{\varrho}(s)$, $\boldsymbol{\varphi}(s)$ and $\boldsymbol{\psi}(s)$. Note that several states can have the same cardinality vectors. The Wants vector of any absorbing state is $1\times M$, which we denote by $\boldsymbol{\psi_0}$. \\

\subsubsection{Action Spaces $\mathcal{A}(s)$}
For each state $s$, the action space $\mathcal{A}(s)$ consists of the set of all possible maximal cliques in the IDNC graph $\mathcal{G}(s)$, constructed from the SFM $\mathbf{F}(s)$.\\

\subsubsection{State-Action Transition Probabilities}
To define the state-action transition probability $P_{\kappa(s)}(s,s')$ for an action $\kappa(s) \in
\mathcal{A}(s)$, we first introduce the following two sets:
\begin{align}
\mathcal{X} &= \left\{i\in\mathcal{T}(\kappa(s))\;\big|\;\varphi_i(s)>\varphi_i(s')\right\}\\
\mathcal{Y} &= \left\{i\in\mathcal{T}(\kappa(s))\;\big|\;\varphi_i(s)=\varphi_i(s')\right\}\;.
\end{align}
The first set includes the targeted receivers whose Lacks sets have decreased from state $s$ to state $s'$, and thus have successfully received the IDNC packet generated from $\kappa(s)$. The second set includes the targeted receivers that have lost the IDNC packet generated from $\kappa(s)$ and thus their Lacks sets did not change. Based on the definitions of these sets, $P_{\kappa(s)}(s,s')$ can be expressed as follows:
\begin{equation} \label{eq:trans-prob}
P_{\kappa(s)}(s,s') = \prod_{i\in\mathcal{X}}\:q_{i} \cdot \prod_{i\in\mathcal{Y}}\:p_{i}\;.
\end{equation}
\fref{fig:ssp-example} depicts the state representation and the action space for the example in \fref{fig:IDNC-graph}. It also depicts the possible transitions given that action $a_7$ is performed.\\
\begin{figure}[t]
\centering
  \includegraphics[width=0.9\linewidth]{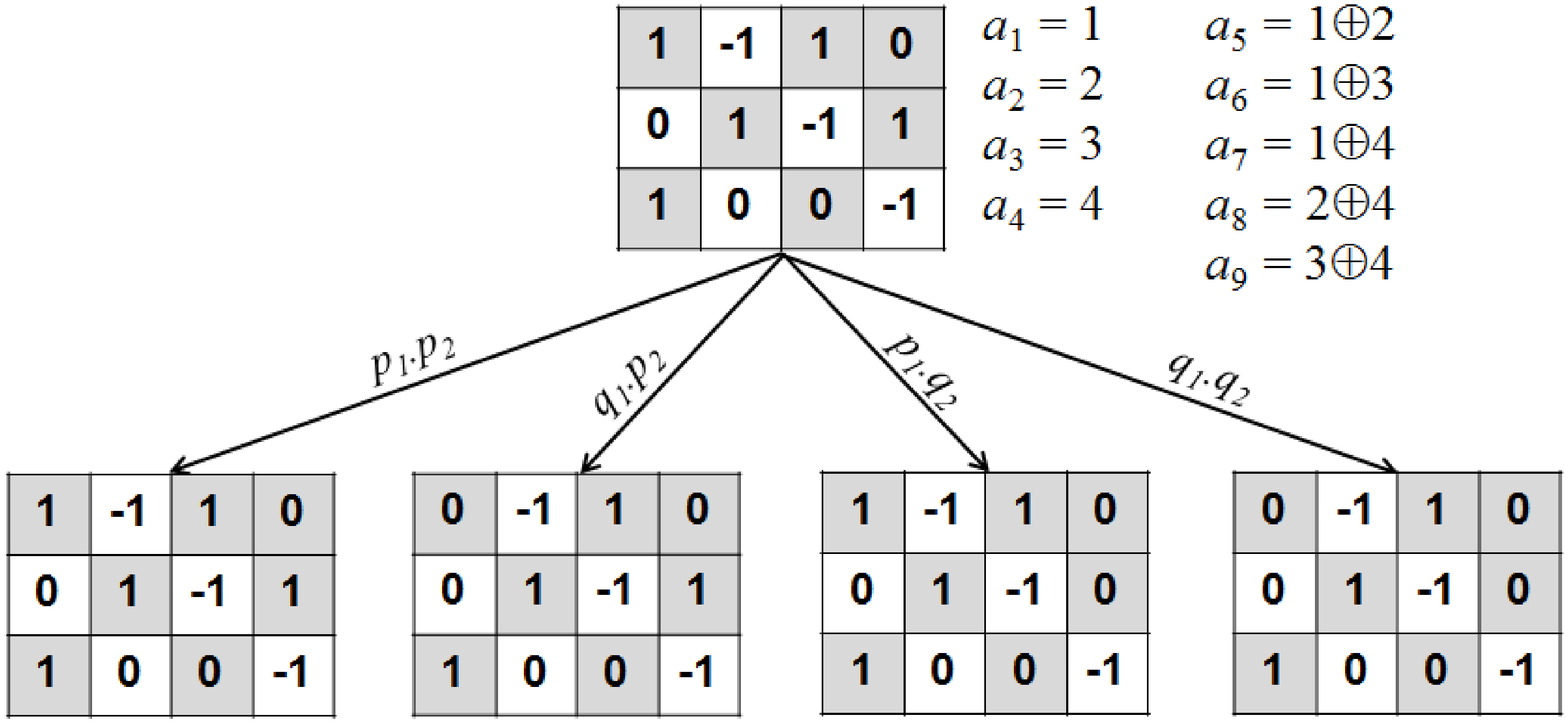}\\
  \caption{State representation, action space of the example in \fref{fig:IDNC-graph} and its possible transitions for action $a_7$.}\label{fig:ssp-example}
\end{figure}

\subsubsection{State-Action Costs}
The expected completion delay is defined in SSP terms as the expected number of transitions in the
process before arriving to an absorbing state. Since any transition (due to any action) takes one
packet transmission, the cost payed by the process is one time-slot. Consequently, the costs of all
actions in all states should be set to 1. In other words, $c(s,\kappa(s)) = 1$ $\forall~\kappa(s)\in\mathcal{A}(s),
s\in\mathcal{S}$.

\subsection{SSP Solution Complexity}
The optimal policy of an SSP problem can be computed using the policy iteration and value iteration algorithms. The complexities of these algorithms are $\Theta\left(|\mathcal{S}|^3 + |\mathcal{S}|^2|\mathcal{A}|\right)$ and $\Theta\left(|\mathcal{S}|^2|\mathcal{A}|\right)$. According to the dimensions of $\mathcal{S}$ and $\mathcal{A}(s)$ described in \sref{sec:formulation}, we conclude that computing the optimal policy is very difficult in real-time for typical values of $M$ and $N$. Even the simulation based technique proposed in \cite{5152148} will not be able to compute the optimal policy in real-time since its complexity still scales with $|\mathcal{S}|$.

\section{SSP Properties} \label{sec:properties}
Despite the complexity of solving the SSP problem formulated in \sref{sec:formulation}, we can study its properties and structure to draw the characteristics of policies that can efficiently minimize the expected completion delay. From \sref{sec:formulation}, it is easy to infer that the SSP formulation has the following properties:
\begin{property}[Uniform Cost] \label{prop:cost}$\quad$\\
The costs of all actions in all states are all the same except for the absorbing state.
\end{property}
\begin{property}[Non-singleton acyclicity] \label{prop:acyclicity}$\quad$\\
No state can be revisited once the process moves to a next state, and thus the SSP formulation is acyclic. Indeed, if some packets are received by some receivers when an action is taken at a given state, there is no means of going back with these receivers not having these packets. However, a state can revisit itself (singleton cycles) if none of the targeted receivers by the taken action receives the IDNC packet.
\end{property}
\begin{property}[Non-increasing successor value functions] \label{prop:decreasing-value}$\quad$\\
Since there are no cycles of size more than one, the successor states of a state $s$ are all closer to the absorbing states than $s$. Consequently, the expected cost to absorption starting from a given state is always greater than or equal to the expected costs to absorption starting from all its successor states.
\end{property}

These three properties can be employed to draw the properties of the optimal policy $\pi^*$ minimizing the mean completion delay at any given state $s$ as follows. From the uniform cost in \prref{prop:cost}, we have:
\begin{align} \label{eq:optimal-policy}
\pi^*(s) &= \arg\min_{\kappa(s)\in\mathcal{A}(s)} \left\{ 1 + \sum_{s'\in\mathcal{S}(s,\kappa(s))} P_\kappa(s,s')\:V_{\pi^{*}}(s')\right\} \nonumber\\
 &= \arg\min_{\kappa(s)\in\mathcal{A}(s)} \left\{\sum_{s'\in\mathcal{S}(s,\kappa(s))} P_\kappa(s,s')\:V_{\pi^{*}}(s')\right\} \nonumber\\
 &= \arg\min_{\kappa(s)\in\mathcal{A}(s)} \left\{\mathds{E}_{\kappa(s)} \left[V_{\pi^{*}}(s')\right]\right\}\;,
\end{align}
where $\mathbb{E}_{\kappa(s)}$ is the expectation operator over the different transition probabilities when
action $\kappa(s)$ is taken. Thus, the optimal action at state $s$ is the action minimizing the expectation
of the optimal value functions of the successor states. From Properties \ref{prop:acyclicity} and \ref{prop:decreasing-value}, we know that all successors of state $s$ are closer to the absorbing state (thus having smaller mean completion delays) except for itself. Consequently, the optimal action at state $s$ is the one that has high probability in moving to states with the minimum expected residual completion delay (i.e. minimum mean time to absorption), given the optimal policy.

Now the problem is that there is no close form expression for the optimal value functions $V_{\pi^*}(s')$ in IDNC and thus there is no means of accurately computing it to determine the optimal policy without solving the SSP. However, based on the previous properties and facts, we can easily infer that the value of $V_{\pi^*}(s')$ for any $s'$, that is successor to state $s$, depends on two main factors:
\begin{itemize}
\item The closeness of the state's Wants vector $\boldsymbol{\psi}(s')$ to that of the absorbing states $\boldsymbol{\psi_0}$.
\item The number and sizes of the primary maximal cliques available as actions in state $s'$.
\end{itemize}
Indeed, the smaller the distance between $\boldsymbol{\psi}(s')$ and $\boldsymbol{\psi_0}$, the smaller the value function of state $s'$. However, this condition is not enough as we should also check the availability of efficient actions at this state that can bring the system faster to an absorbing state. In general, the successor states of $s$, whose primary graphs include more numerous and larger maximal cliques, have more chances of reaching the absorbing state faster than the others. Since all states $s'$ are successors of a same state $s$, their graphs are different variants of $\mathcal{G}(s)$ depending on the vertices that have been served. Consequently, the action at state $s$, which can maximize the coding opportunities in the IDNC primary graph at state $s'$, will result in larger and more numerous primary maximal cliques, which will help in reaching an absorbing state faster in future steps. Based on these observation, we state that the policy that can efficiently reduce the expected completion delay in IDNC should aim, at any visited state, to both:
\begin{itemize}
\item Bring the system Wants vector the closest to the absorbing states vector $\boldsymbol{\psi_0}$.
\item Maximize the coding opportunities in the successor state's primary graph.
\end{itemize}
If we can find a policy that can simultaneously achieve these two goals, we will employ it to design an algorithm to efficiently reduce the expected completion delay for IDNC. To investigate the existence of such policy, we will study two important features of the problem, namely its geometric structure, and the evolution of coding opportunities in the IDNC graph. This will be the target of the next two sections.

\section{Geometric Structure} \label{sec:geometry}
In this section, we will explore the actions, which have high chances of moving the system Wants vector closest to that of the absorbing states. Given the representation of the SSP states by their Wants vector, we can define a geometric structure as follows. Define an $M$-dimensional space, and locate to each point $\Psi = \left[\Psi_1,\dots,\Psi_M\right]$ in this space all the states having Wants vectors equal to the coordinates of this point. Although many states can share the same Wants vector and thus can be located at the same point, these states differ from one another by their IDNC graphs. All absorbing states will be located at the origin $\Psi_0$ of this space. Note that this geometric representation has the same non-singleton acyclicity property as the SSP (i.e. a point cannot be revisited after it is left).\ignore{ In this representation, we can express the actions as vectors of dimension $M$, with the $i$-th entry equal to one (zero) if receiver $i$ is (not) targeted by the action.}

Since at most one packet can be decoded by each receiver from any IDNC transmission, the system can at most move from the point $\Psi = \boldsymbol{\psi}(s)$ to a point $\Psi' = \boldsymbol{\psi}(s')$ which is a vertex in the hypercube $\Gamma(s)$ defined as:
\begin{equation}
\Gamma(s) = \bigg\{\Psi' \: \big| \: \psi_i(s) - \psi_i(s') \in \{0,1\} ~~\forall~i\in\mathcal{M}\bigg\}\;.
 \end{equation}
In other words, $\Gamma(s)$ is the hypercube of side length 1, in which $\psi_i(s)$ is the corner having the largest coordinates. In this case, the optimal action at any state is the one that can transitions the system to the opposite diagonal point in the M-dimensional hypercube, for which $\psi_i(s) - \psi_i(s') = 1~\forall~i\in\mathcal{M}$. This action means that all $M$ receivers are targeted with primary packets. If such actions exist and are applied in all visited states, we will reach completion faster. However, these actions will most probably not exist in most states due to the instantly decodability constraint. Consequently, We need a method to estimate the closeness of other points to the absorbing state to evaluate our IDNC scheme.

\fref{fig:geo-structure-1} depicts the geometric structure of the example in \fref{fig:ssp-example} after removing the fourth column (i.e. removing the fourth packet and the actions it appears in). Consequently the system is at point identified by the Wants vector $\boldsymbol{\psi}(s) = [2,1,1]$.
\begin{figure}[t]
\centering
\includegraphics[width=0.6\linewidth]{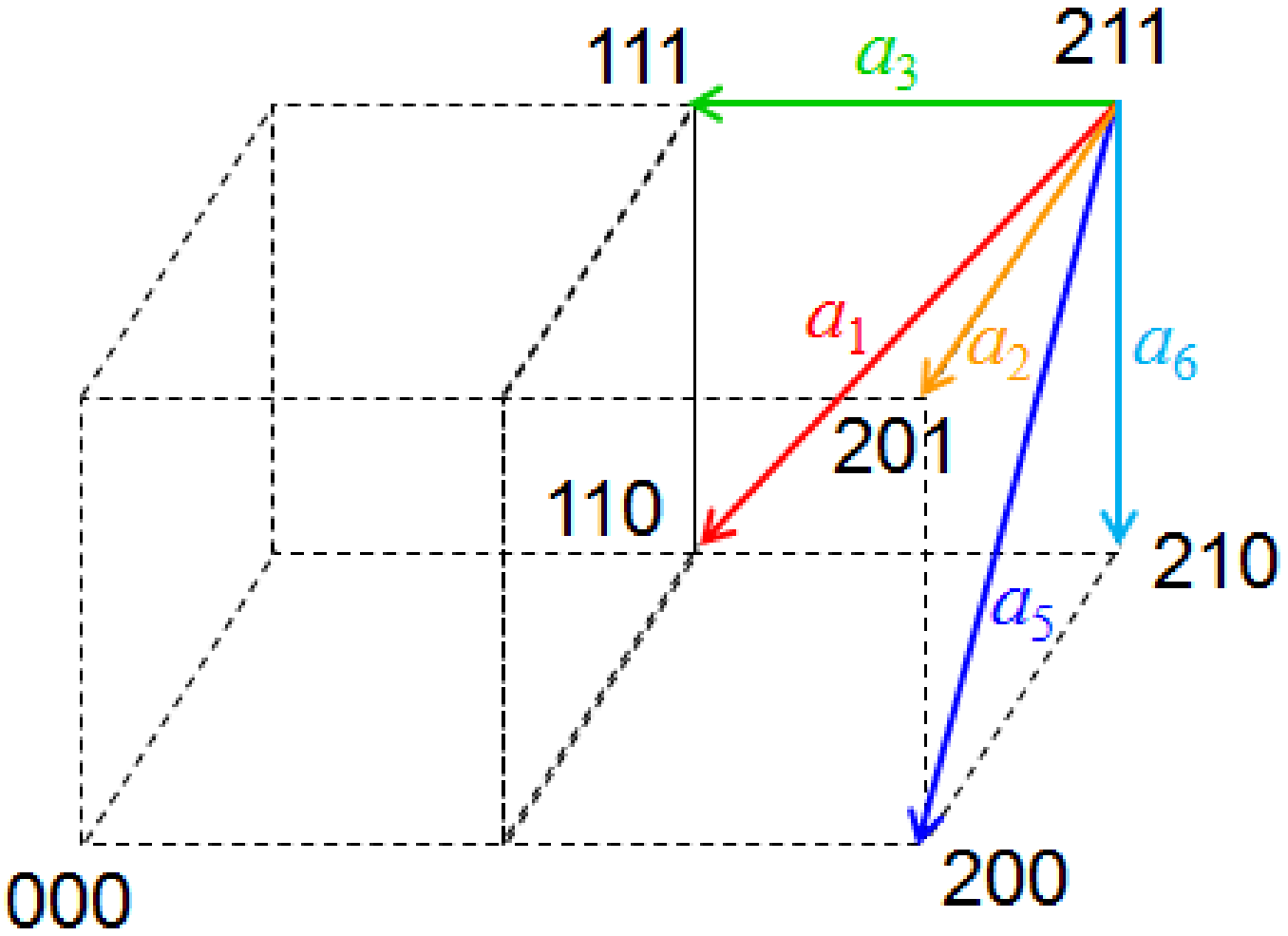}\\
\caption{Geometric structure of the example in \fref{fig:ssp-example} after removing the fourth column}
\label{fig:geo-structure-1}
\end{figure}
In this example, there are only five actions $a_1, a_2, a_3, a_5, a_6$ according to their notation in \fref{fig:ssp-example}. Assuming that the system is erasure-free, action $a_3$ will lead the system to point $[1,1,1]$ whereas action $a_5$ will lead it to point $[2,0,0]$. Although $a_5$ targets more receivers than $a_3$ (which makes it a capacity achieving action as of \cite{Wang2010,Gatzianas2010}), we can clearly see that the latter action gives the chance to the system to reach absorption with one more erasure-free transmission, if there exists an IDNC packet targeting all three receivers at the state located at point $[1,1,1]$. This closeness to absorption in terms of Wants vectors is shown through the smaller geometric distance from point $[1,1,1]$ to $\Psi_0$, compared to point $[2,0,0]$. We can infer from this example that minimizing the maximum entry of the Wants vector (i.e. $\max_i\{\psi_i(s')\}$) brings the system closest to the absorbing point. The intuition behind this finding is that the receivers having the largest Wants sets will impose their Wants set cardinalities as lower bounds on the completion delay. Consequently, serving these receivers first gives hope to reduce this lower bound at each step whereas ignoring them will not change the lower bound.

However, minimizing the maximum of the Wants vector entries is not enough to describe the actions with closest successor states to absorption. For example, actions $a_1$ and $a_3$ in \fref{fig:geo-structure-1} have the same value for $\max_i\{\psi_i(s')\}$, but $a_1$ brings the system closer to the absorbing point compared to $a_3$ in terms of Wants vector, since it serves an additional receiver with smaller Wants set. This is also reflected on the geometric distance from the two destination points to $\Psi_0$.

From the above example, we can conclude that, in order to bring the system closest to the absorbing point in terms of Wants vector, the sender should give more weight to serving the receivers with largest coordinate entries, while maximizing the number of served receivers with the smaller coordinate entries. This weighting can be done through norm expressions. For example, the $L_2$ norm (Euclidian distance) represented the state closeness to absorption, in terms of the Wants vectors, in the previous examples. The larger the employed norm, the more biased the weighting in giving service to the receivers with largest coordinate entries.

For erasure channels, the effect of packet erasures should be reflected on the geometric structure of the problem. Let $i$ and $k$ be two receivers having the same Wants set size but $p_i > p_k$. Consequently, $i$ will require on average more targeting attempts compared to $k$ in order to deplete its Wants set. Since we assume that erasure probabilities do not change during the transmission of a frame, targeting $k$ and ignoring $i$ is expected to result in a higher overall completion delay, especially when $\widetilde{\psi}_i(s)$ is among the largest value in $\boldsymbol{\widetilde{\psi}}(s)$. According to these facts and the intuition explained above for the erasure-free case, $i$ should be given a higher priority of service than $k$.\ignore{ to reduce its higher impact on the completion delay. This can be done by representing the $i$-th coordinate further from $\Psi_0$ than that of the $k$-th when using weighting through norms.}

To implement the above prioritization, we define a channel weighted Wants vector $\boldsymbol{\widetilde{\psi}}(s) = \left[\widetilde{\psi}_1(s),\dots,\widetilde{\psi}_M(s)\right] $, where $\widetilde{\psi}_i(s) = \frac{\psi_i(s)}{q_i}$. Based on this new vector definition, we can redefine our space such that its points $\Psi$ are identified by the coordinates of the vectors $\boldsymbol{\widetilde{\psi}}(s)$ instead of $\boldsymbol{\psi}(s)~\forall~s\in\mathcal{S}$. In this case, the actions move the system within hyper-rectangles $\Gamma'(s)$ with sides equal to $q_i^{-1}$ in the $i$-th dimension. In other words:
\begin{equation}
\Gamma'(s) = \bigg\{\Psi' \: \big| \: \psi_i(s) - \psi_i(s') \in \{0,q_i^{-1}\} ~~\forall~i\in\mathcal{M}\bigg\}\;.
 \end{equation}
The sender should then take the action that can reach successor states with minimum $L_n$ norm over this geometry.

From the above observations, we can draw a conclusion that the policies, which efficiently reduce the IDNC completion delay, should always aim at each visited state $s$ to reach a state $s'$ that is located at the point with minimum distance to the absorbing point $\Psi_0$ and thus the minimum $\left\|\widetilde{\boldsymbol{\psi}}_i(s')\right\|_n$. Consequently, the receivers with larger values of $\widetilde{\psi}_i$ will have higher priority to be selected for transmission at state $s$. Now, if we can show that this norm based selection of the receivers also maximizes the coding opportunities in the successor states, then this norm based selection policy is our searched policy as explained in \sref{sec:properties}. To investigate this point, we will study the evolution of the coding opportunities in the IDNC graph with respect to the selection of receivers in each transmission in the next section.

\section{Evolution of Coding Opportunities} \label{sec:graph-evolution}
As stated in \sref{sec:properties}, one major factor that identifies the efficiency of an action in reducing the completion delay is its ability to maximize the coding opportunities in the primary IDNC graph. We know from the IDNC graph structure that the coding opportunities are represented by its edges. Consequently, the overall number of coding opportunities in the graph are represented by its edge set size.\ignore{ Since the graphs of all the successors $s'$ of a state $s$ are different variants of $\mathcal{G}(s)$, the action at state $s$ that maximizes the expected edge set size in the primary graph will end up at the successor state with more opportunities to reach an absorbing state faster in future steps.} In \cite{TON10-CE}, we studied the receiver selection strategies, maximizing the coding opportunities and density in the IDNC graph for wireless broadcast. This study is done based on ignoring the packet content of the Has, Lacks and Wants sets and computing the expected edge set size given their cardinalities. In this paper, we will extend this study to the case of wireless multicast by first introducing the following theorem\ignore{, presenting an expression for the expected evolution of the edge set size of the primary graph when an arbitrary clique $\kappa$ is chosen for transmission}. In all upcoming analysis in this section, we assume the system is at an arbitrary state $s$ and all the variables represent their values at this state, thus dropping the $(s)$ notation for simplicity.

\begin{theorem} \label{th:graph-density-evolution}
For a given maximal clique $\kappa$, chosen for transmission at time $t$ in a multicast scenario, the expected edge set cardinality of the IDNC primary graph at time $t+1$ is expressed as:
\begin{align}\label{eq:graph-density-evolution}
\mathds{E}\Bigg[\left|\mathcal{E}_\rho^{(t+1)}\right|\Bigg] = \;\;& \mathds{E}\Bigg[\left|\mathcal{E}_\rho^{(t)}\right|\Bigg]- \frac{1}{2} \sum_{i\in\mathcal{T}_\rho(\kappa)} q_i \left(\mathds{E}\left[\Delta_i^{(t)}\right] + \gamma_i\right)  \nonumber\\
& + \frac{1}{2}\sum_{i\in\mathcal{T}(\kappa)}  \psi_i \alpha_i + \frac{1}{2}\sum_{i\notin\mathcal{T}(\kappa)}  \psi_i \beta_i\;,
\end{align}
where
\begin{align}
&\alpha_i = \sum_{\substack{k=1\\k\neq i}}^M q_i \xi_k -\sum_{\substack{k\in\mathcal{T}_\rho(\kappa)\\k\neq i}}\Phi_{ik}(q_i)  + \sum_{\substack{k\in\mathcal{T}_\sigma(\kappa)\\k\neq i}} \Lambda_{ik}(q_i) \;,\label{eq:alpha}\\
&\beta_i = -\sum_{\substack{k\in\mathcal{T}_\rho(\kappa)\\k\neq i}}\Phi_{ik}(0)  + \sum_{\substack{k\in\mathcal{T}_\sigma(\kappa)\\k\neq i}} \Lambda_{ik}(0) \;,\label{eq:beta}\\
&\gamma_i = \sum_{\substack{k=1\\k\neq i}}^M \xi_k -\sum_{\substack{k\in\mathcal{T}_\rho(\kappa)\\k\neq i}}\Phi_{ik}(1)  + \sum_{\substack{k\in\mathcal{T}_\sigma(\kappa)\\k\neq i}} \Lambda_{ik}(1)\;,\label{eq:gamma}\\
&\Phi_{ik}(x) = \frac{q_k}{N} \left( 1 + \frac{\left(\varrho_k-\psi_k+1\right) \left(\varrho_i+x\right)}{N-1}\right) \;,\label{eq:phi}\\
&\Lambda_{ik}(x)  = \frac{q_k \psi_k\left(\varrho_i+x\right)}{N(N-1)}\;,\qquad\xi_k = \frac{\psi_k \varrho_k}{N(N-1)}\;,
\end{align}
and $\mathds{E}\left[\Delta_i^{(t)}\right]$ is the expected degree of a vertex of receiver $i$ at time $t$.
\end{theorem}
\begin{proof}
The proof can be found in \appref{app:graph-density-evolution}.
\end{proof}
From the above theorem, we can draw the following observations about the maximization of the edge set size.

\subsection{Causal Evolution}
The first term in the right-hand side of \eqref{eq:graph-density-evolution} is the expected edge set cardinality at time $t$. This means that the edge set size at one time depends on its previous size instances and thus the evolution process is causal. Consequently, if the edge set cardinality at $t+1$ was maximized, the future evolutions in visited successor states at times $t+2$ and $t+3$, ... will also benefit from the maximization that occurred at time $t$, which results in the maximization of their edge set size, if the same policy is employed in each step.

\subsection{Vertex Disappearance}
The second term in \eqref{eq:graph-density-evolution} represents an expected reduction in the edge set size due to the possible disappearance of the primary targeted vertices. This disappearance results in the removal of their adjacent edges at time $t$, which is reflected in the $\mathds{E}\left[\Delta_i^t\right]$ term. It also results in the loss of the potential improvement in these degrees if they were kept in the graph, which is reflected in the $\gamma_i$ term. This loss is a natural outcome of the recovery transmission process and is unavoidable. We cannot try to reduce this term by reducing the size of the primary targeted receiver set as this will tend to increase the expected completion delay. However, we can still reduce the effect of this loss component by serving the vertices with smaller degrees. The following theorem compares the expected vertex degrees of two receivers given the sizes of their Has and Wants sets.
\begin{theorem} \label{th:degree-comparison}
If $\psi_i > \psi_h$ and $\varrho_i<\varrho_h$, then $\mathds{E}\left[\Delta_h\right] > \mathds{E}\left[\Delta_i\right]$.
\end{theorem}
\begin{proof}
The proof can be found in \appref{app:degree-comparison}
\end{proof}
Now, if $q_i<q_h$ and $\psi_i>\psi_h$, $\frac{1}{2} q_i\mathds{E}\left[\Delta_i\right] < \frac{1}{2} q_h\mathds{E}\left[\Delta_h\right]$. Consequently, serving receivers with largest Wants sets and erasure probabilities, and smallest Has sets, results in a smaller loss in the resulting edge set size.

\subsection{Degrees of Remaining Vertices}
The third and fourth terms in \eqref{eq:graph-density-evolution} represent the change in the degrees of the remaining vertices in the primary graph, which are quantified by $\psi_i\alpha_i$ and $\psi_i\beta_i$ for the targeted and non-targeted receivers, respectively. The following theorem describes the relation between these two terms.
\begin{theorem} \label{th:alpha-beta}
The increase in the degrees of the remaining vertices of any receiver is larger when it is targeted than when it is not. In other words, $\alpha_i  \geq \beta_i~\forall~i\in\mathcal{M}$.
\end{theorem}
\begin{proof}
The proof can be found in \appref{app:alpha-beta}.
\end{proof}
Now, moving a receiver $i$ from the non-targeted set to the targeted set results in an extra number of edges of $\frac{\psi_i}{2} \left(\alpha_i - \beta_i\right)$. This term is larger when $\psi_i$ is larger, and thus moving a receiver with a larger Wants set to the targeted receiver set adds more edges to the primary graph than moving a receiver with a smaller Wants set. Consequently, a larger increase in the expected edge set size is obtained when targeting the maximum number of receivers having larger Wants sets.
\ignore{since only one of these parameters for any receiver will be reflected in all its vertices in the primary graph (which appears as a multiplication with its Wants set cardinality in \eqref{eq:graph-density-evolution}), a larger increase in the expected edge set evolution is obtained when targeting (i.e. multiplying by $\alpha_i$) the maximum number of receivers having larger Wants sets.}

Another important insight about the values of $\alpha_i$ and $\beta_i$ can be inferred from the analysis of their components $\Phi_{ik}$ and $\Lambda_{ik}$. Since the terms $\sum_{\substack{k\in\mathcal{T}_\rho(\kappa)\\k\neq i}}\Phi_{ik}(q_i)$ and $\sum_{\substack{k\in\mathcal{T}_\rho(\kappa)\\k\neq i}}\Phi_{ik}(0)$ are subtractive terms from $\alpha_i$ and $\beta_i$, respectively, then selecting the receivers with smaller values of $\Phi_{ik}(q_i)$ and $\Phi_{ik}(0)$ to be primary targeted receivers increases the values of $\alpha_i$ and $\beta_i$, respectively. Now, if $q_k<q_h$, $\psi_k > \psi_h$ and $\varrho_k < \varrho_h$, we have:
\begin{align}
& q_k\left(\varrho_k-\psi_k+1\right) < q_h\left(\varrho_h-\psi_h+1\right) \\
\Rightarrow \quad &\Phi_{ik}(q_i) < \Phi_{ih}(q_i) \quad\mbox{and}\quad \Phi_{ik}(0) < \Phi_{ih}(0)
\end{align}
Consequently, the receivers having larger Wants sets and erasure probabilities, and smaller Has sets, have smaller values of $\Phi_{ik}(q_i)$ and $\Phi_{ik}(0)$. In case of equal demand ratios at all receivers (including the broadcast scenarios), the three above conditions are equivalent on average. In other words, the receivers having smaller reception probabilities will on average have larger Wants sets and smaller Has sets. Consequently, maximizing the number of such receivers in the set of primary targeted receiver maximizes the values of both $\alpha_i$ and $\beta_i$.

Since the terms $\sum_{\substack{k\in\mathcal{T}_\sigma(\kappa)\\k\neq i}}\Lambda_{ik}(q_i)$ and $\sum_{\substack{k\in\mathcal{T}_\sigma(\kappa)\\k\neq i}}\Lambda_{ik}(0)$ are additive terms to $\alpha_i$ and $\beta_i$, respectively, then selecting the receivers with larger values of $\Lambda_{ik}(q_i)$ and $\Lambda_{ik}(0)$ to be secondary targeted receivers increases the values of $\alpha_i$ and $\beta_i$, respectively. The values of $\Lambda_{ik}(q_i)$ and $\Lambda_{ik}(0)$ are larger for receivers having larger value of $q_k\psi_k$ and thus having larger Wants sets but lower erasure probabilities.\ignore{ Since $\mathcal{T}_\rho$ is always selected before $\mathcal{T}_\sigma$, and since $\mathcal{T}_\rho$ should include the receivers with having the largest Wants sets and erasure probabilities (with low $\Phi_{ik}(q_i)$ and $\Phi_{ik}(0)$), then the receivers selected for $\mathcal{T}_\sigma$ (i.e. with high $\Lambda_{ik}(q_i)$ and $\Lambda_{ik}(0)$ values), are usually the remaining receivers with largest Wants sets after selecting $\mathcal{T}_\rho$.
that fall in the intermediate ranges of $q_k$ and $\psi_k$. In our algorithm design, we will assume that these receivers are the remaining receivers with largest Wants sets and erasure probabilities, after targeting the receivers having the largest Wants sets and erasure probabilities with primary packets. Although this assumption may not be always valid, we will show in \sref{sec:simulations} that it is sufficient to achieve an outstanding completion delay performance.}

\subsection{Overall Maximization Strategy}
From the above theorems and discussion, we can infer that, after a given transmission, the value of the edge set size in the primary graph is maximized\ignore{ vertex disappearance effect is reduced and the remaining vertex degrees are maximized} by targeting the receivers, having the largest Wants sets and erasure probabilities (i.e. largest $\widetilde{\psi}_i$ values), with primary packets, then targeting the remaining receivers, having the largest Want sets and lower erasure probabilities, with secondary packets. We will refer to this strategy as \emph{the worst receiver layered targeting (WoRLT)} strategy.\ignore{ This strategy both maximizes the coefficients $\alpha_i$ and $\beta_i$ for all $i\in\mathcal{M}$ and multiplies the maximum number of largest $\psi_i$ values with their larger coefficients $\alpha_i$. It also removes the vertices with smaller degrees from the graph. All these effects of the WoRLT strategy maximizes the expected coding opportunities at time $t+1$.}

In the primary graph, the WoRLT strategy is equivalent to a norm minimization
of the channel weighted Wants vector $\boldsymbol{\widetilde{\psi}}$. Indeed, such minimization will result in targeting the maximum number of receivers having the largest Wants sets and erasure probabilities. According to the discussion in \sref{sec:geometry}, this policy perfectly matches the policy bringing the system the closest to the absorbing point $\Psi_0$.

For the secondary graph, the WoRLT strategy is still equivalent to a norm minimization
of the channel weighted Wants vector $\boldsymbol{\widetilde{\psi}}$ within the IDNC secondary subgraph, adjacent to all selected vertices in the primary graph. Since the receivers with the largest Wants sets and erasure probabilities will be targeted with vertices from the primary graph, and since each receiver can have at most one vertex per clique, applying the norm minimization in the secondary subgraph, adjacent to the primary selected vertices, will result in targeting the remaining receivers with largest Wants sets and lower erasure probabilities. This step does not conflict with the policy bringing the process closest to absorption but rather fosters it in future steps. Indeed, serving these receivers with secondary packets will increase the coding opportunities of their remaining primary (and secondary) vertices, which allows them to get served faster in future transmissions, thus bringing the process closest to absorption.

Given the above facts, we conclude this section by stating that the WoRLT strategy is the policy that can efficiently reduce the IDNC completion delay, as claimed in \sref{sec:properties}. We will thus design our proposed algorithms according to this strategy in the following section.

\section{Proposed Algorithms} \label{sec:algorithms}

\subsection{Maximum Weight Clique Selection Algorithm} \label{sec:MWCS}
According to the findings of the previous sections, we propose a two-step maximal clique selection algorithm that should be executed at any visited state $s$. In the first step, the algorithm selects the maximal clique $\kappa_\rho^*(s)$ in the primary graph that targets receivers with larger channel weighted Want set sizes, thus minimizing $\|\widetilde{\psi}_i(s)\|_n$ for the expected successor state and maximizing the number of edges in its graph. To further maximize the number of edges, the same process should be done for the secondary subgraph connected to $\kappa_\rho^*$ to find $\kappa_\sigma^*(s)$. Each of these two steps can be done using a maximum weight clique selection algorithm as follows.

For each vertex $v_{ij}$ in the multicast IDNC graph, we assign a weight $\left(\widetilde{\psi}_i(s)\right)^n$, where $n$ is the order of the selected norm.\ignore{ When a maximum weight clique selection algorithm is executed on the primary graph, the selected maximum weight clique $\kappa_\rho^*$ will target receivers with larger $\left(\widetilde{\psi}_i(s)\right)^n$, thus minimizing $\|\widetilde{\psi}_i(s)\|_n$ for the expected successor state and maximizing the number of edges in its primary graph.} After finding this clique, the secondary subgraph connected to $\kappa_\rho^*$ is extracted and the maximum weight clique selection algorithm is run on it to obtain $\kappa_\sigma^*$. When both cliques are found, the sender sends an IDNC packet that is generated by XORing all the source packets identified by the vertices in both cliques. After receiving the feedback from the receiver, the sender determines the reached successor state and the whole procedure is re-executed. This loop is run until all vertices in the primary graph are depleted.

It is well known that the maximum weight clique selection problem is NP-hard \cite{Garey1979}, and is hard to approximate \cite{Ausiello1999}. On the other hand, there exist several algorithms that solve this problem in polynomial time for moderate size graphs (\cite{Yamaguchi2008} and references therein). However, the complexity of these algorithms may still be prohibitive for the applications of interest in this paper \cite{Yamaguchi2008}. Consequently, we will design a simple heuristic in the next section to solve the problem with much lower complexity.

\subsection{Maximum Weight Vertex Search Algorithm} \label{sec:MWVS}
In this section, we design a simple algorithm that performs clique selection, using a maximum weight vertex search. For this search to be efficient, the vertices' weights must not only reflect the $\left(\widetilde{\psi}_i(s)\right)^n$ values of their inducing receivers, but also their adjacency to vertices having high $\left(\widetilde{\psi}_i(s)\right)^n$.

To design the vertices' weights, we first define $a_{ij,kl}(s)$ as the adjacency indicator of vertices $v_{ij}$ and $v_{kl}$ in $\mathcal{G}(s)$ such that:
\begin{equation} \label{eq:adjacency}
a_{ij,kl}(s) =
\begin{cases}
1 \quad & v_{ij}~\mbox{is connected to}~v_{kl}~\mbox{in}~\mathcal{G}(s)\\
0 \quad &\mbox{otherwise}\;.
\end{cases}
\end{equation}
We then define the weighted degree $\Delta_{ij}(s)$ of vertex $v_{ij}$ as:
\begin{equation} \label{eq:weighted-degree}
\Delta_{ij}^{w}(s) = \sum_{\forall v_{kl}\in\mathcal{G}(s)}\:a_{ij,kl}(s)\: \left(\widetilde{\psi}_k(s)\right)^n\;.
\end{equation}
Thus, a large weighted vertex degree reflects its adjacency to a large number of vertices
belonging to receivers with large values of $\left(\widetilde{\psi}_i(s)\right)^n$. We finally define the vertex weight
$w_{ij}(s)$ as:
\begin{equation} \label{eq:weights}
w_{ij}(s) = \left(\widetilde{\psi}_i(s)\right)^n\: \Delta_{ij}^w(s) \;.
\end{equation}
Consequently, a vertex $v_{ij}$ has a large weight when it both belongs to a receiver with large $\left(\widetilde{\psi}_i(s)\right)^n$ value and is adjacent to a large number of vertices with large $\left(\widetilde{\psi}_k(s)\right)^n$ values.\ignore{ Let $\mathcal{G}^{v}(s)$ be the subgraph in $\mathcal{G}(s)$ only including all vertices adjacent to vertex $v$. We finally define $\mathbf{A}(s)$ and $\mathbf{A}_{v}(s)$ as the adjacency matrices of $\mathcal{G}(s)$ and $\mathcal{G}_{v}(s)$, respectively.}

Based on these definitions, we can introduce our proposed packet selection algorithm as follows. The algorithm operates only for visited states. In each visited state $s$, the algorithm first computes a primary maximal clique $\kappa_\rho^*(s)$ in $\mathcal{G}_\rho(s)$. At first, $\kappa_\rho^*(s)$ and $\kappa_\sigma^*(s)$ are empty sets. The algorithm starts by selecting the maximum weight vertex in $\mathcal{G}_\rho(s)$ to be the source vertex in $\kappa^*(s)$. For each of the following iterations,
the algorithm first recomputes the new vertex weights within the primary subgraph connected to all previously selected vertices in $\kappa_\rho^*(s)$, then adds the new maximum weight vertex to it. When there is no further primary vertices adjacent to all vertices in
$\kappa_\rho^*(s)$, the same process is repeated with the secondary subgraph adjacent to $\kappa_\rho^*$ until no vertices are remaining in the global graph. The final maximal clique $\kappa^*(s)$ is thus the union of $\kappa_\rho^*(s)$ and $\kappa_\sigma^*(s)$. Once this clique is computed, the sender forms and sends an IDNC packet by XORing the source packets identified by the vertices in $\kappa^*(s)$. According to the received feedback, a new state is visited and the process is re-executed until the absorbing state is reached.\ignore{ The pseudo code for the algorithm is depicted in \algref{alg:MWVS}.} Since a maximal clique can have at most $M$ vertices (we can target each receiver at most once per transmission), and since each iteration in the algorithm requires weight computations for the $O(MN)$ graph vertices, the complexity of the algorithm is $O(M^2N)$.
\ignore{
\begin{algorithm}[t]
\begin{algorithmic}
\REQUIRE $\mathbf{F}(s)$ and $\boldsymbol{\widetilde{\psi}}(s)$
\STATE Initialize $\mathcal{K}^*(s) = \varnothing$.
\STATE Construct $\mathcal{G}(s)$ and $\mathbf{A}(s)$.
\WHILE{$\mathcal{G}_\rho(s) \neq \varnothing$}
\STATE Compute $\Delta^w_{ij}(s)$ and $w_{ij}(s)$ using \eqref{eq:weighted-degree} and \eqref{eq:weights}.
\STATE Select $v^* = \arg\max_{v_{ij}\in\mathcal{G}_\rho(s)} \left\{w_{ij}(s)\right\}$.
\STATE Add $v^*$ to $\kappa_\rho^*(s)$
\STATE Set $\mathcal{G}(s)~\leftarrow~\mathcal{G}^{v^*}(s)$ and
$\mathbf{A}(s)~\leftarrow~\mathbf{A}^{v^*}(s)$
\ENDWHILE
\WHILE{$\mathcal{G}(s) \neq \varnothing$}
\STATE Compute $\Delta^w_{ij}(s)$ and $w_{ij}(s)$ using \eqref{eq:weighted-degree} and \eqref{eq:weights}.
\STATE Select $v^* = \arg\max_{v_{ij}\in\mathcal{G}(s)} \left\{w_{ij}(s)\right\}$.
\STATE Add $v^*$ to $\kappa_\sigma^*(s)$
\STATE Set $\mathcal{G}(s)~\leftarrow~\mathcal{G}^{v^*}(s)$ and
$\mathbf{A}(s)~\leftarrow~\mathbf{A}^{v^*}(s)$.
\ENDWHILE
\STATE $\kappa^*(s) = \kappa_\rho^*(s) \cup \kappa_\sigma^*(s)$.
\end{algorithmic}
\caption{Maximum Weight Vertex Search Algorithm} \label{alg:MWVS}
\end{algorithm}
}

\section{Simulation Results} \label{sec:simulations}
In this section, we present simulation results comparing the performance of our proposed algorithms to the following algorithms in both multicast and broadcast scenarios:
\begin{itemize}
\item Random clique search algorithm (RND), employed in \cite{ISIT09}, which picks a random clique from the graph for each transmission .
\item Maximum clique selection algorithm (MC), which selects a primary maximum clique $\kappa_\rho^{max}(s)$ in $\mathcal{G}_\rho$ then selects the secondary maximum clique from the secondary subgraph connected to $\kappa_\rho^{max}(s)$ (denoted by MC in figures).
\item Prefect RNC, in which we assume full independence between all transmitted coding coefficient vectors. Thus, this scheme represents the global optimal completion delay in the broadcast scenario.
\end{itemize}
For our proposed algorithm, we consider the $L_1$, $L_2$, $L_3$, $L_5$ and $L_{10}$ norms to test the effect of the selection bias on the algorithms' performance. In our simulations, we assume that different receivers have different packet erasure probabilities and different demand ratios, that change from frame to frame while keeping the average erasure probability ($p$) and average demand ratio ($\mu$) constant.

Figures \ref{fig:D-opt} \ref{fig:M-opt}, \ref{fig:N-opt} and \ref{fig:P-opt} depict the average completion delay performance of the maximum weight clique selection algorithm with different norms and compare it to the different algorithms against the average demand ratio $\mu$ (for $M = 60$, $N=30$, $p=0.15$), the number of receivers $M$ (for $\mu = 0.5$ and $1$, $N=30$, $p=0.15$), the number of packets $N$ (for $\mu = 0.5$ and $1$, $M=60$, $p=0.15$), and the average erasure probability $p$ (for $\mu = 0.5$ and $1$, $M=60$, $N=30$), respectively.

\begin{figure}[t]
\centering
  \includegraphics[width=0.95\linewidth]{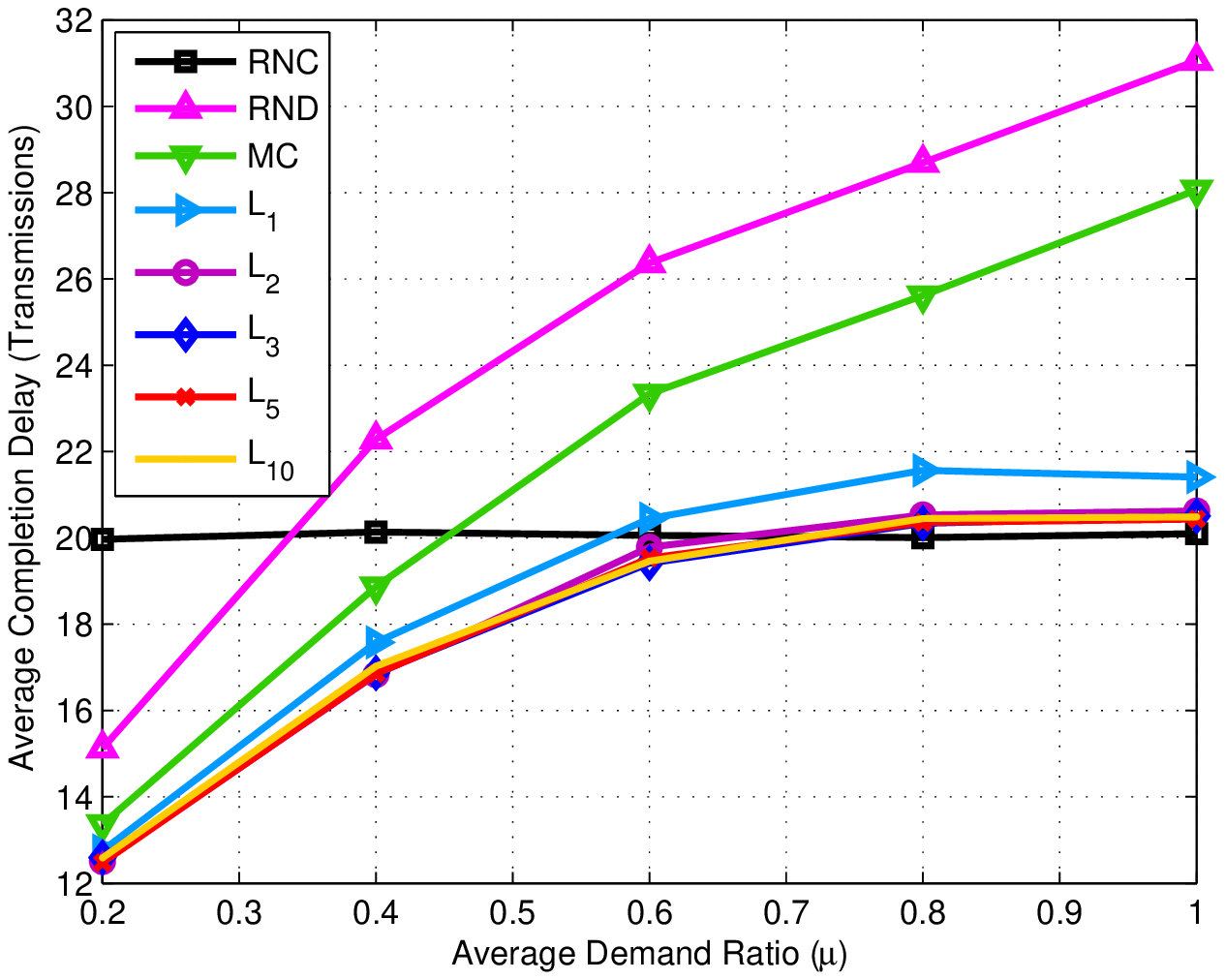}\\
  \caption{Performance comparison of different schemes vs $\mu$}\label{fig:D-opt}
\end{figure}

\begin{figure}[t]
\centering
  \includegraphics[width=0.95\linewidth]{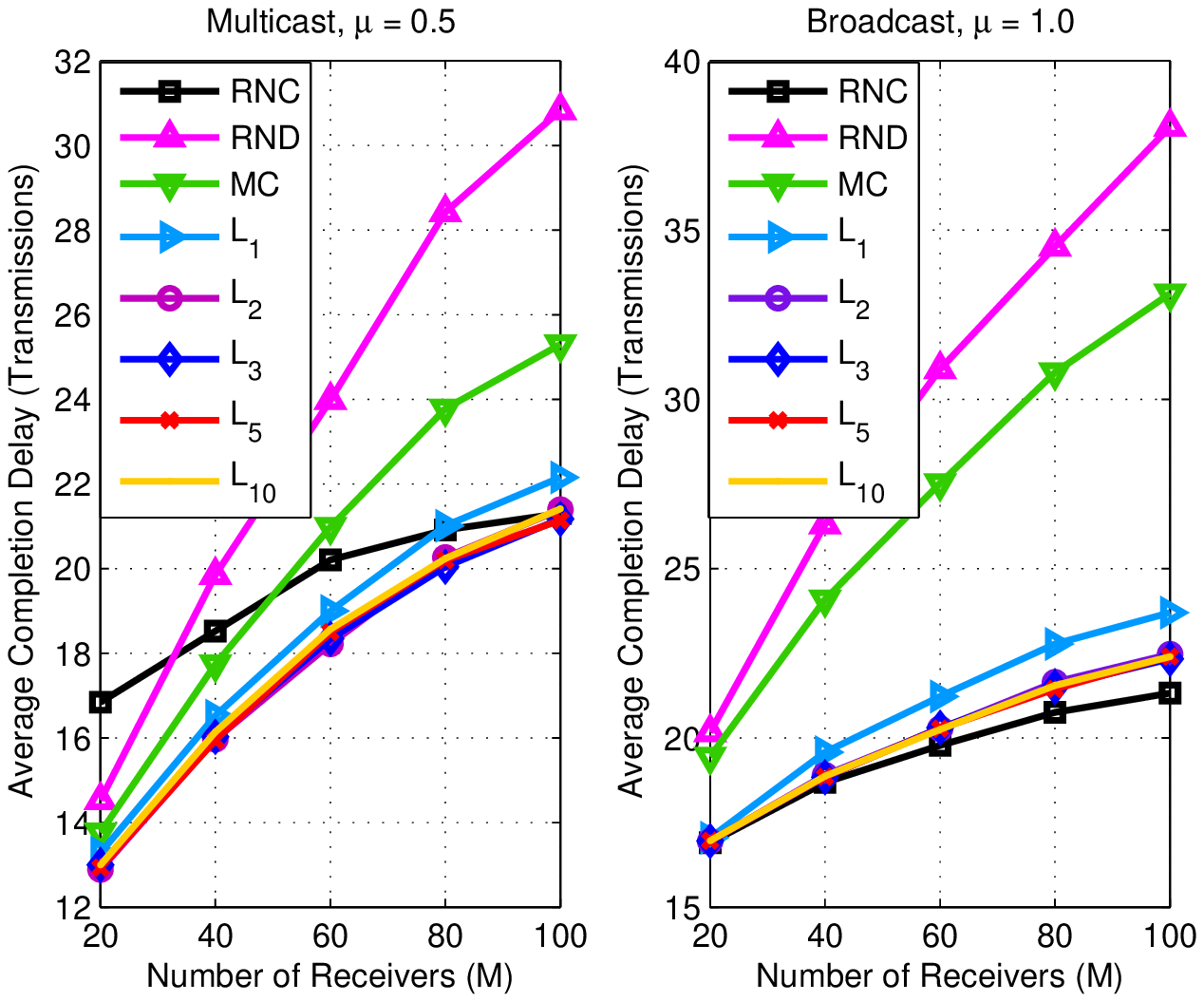}\\
  \caption{Performance comparison of different schemes vs $M$}\label{fig:M-opt}
\end{figure}

\begin{figure}[t]
\centering
  \includegraphics[width=0.95\linewidth]{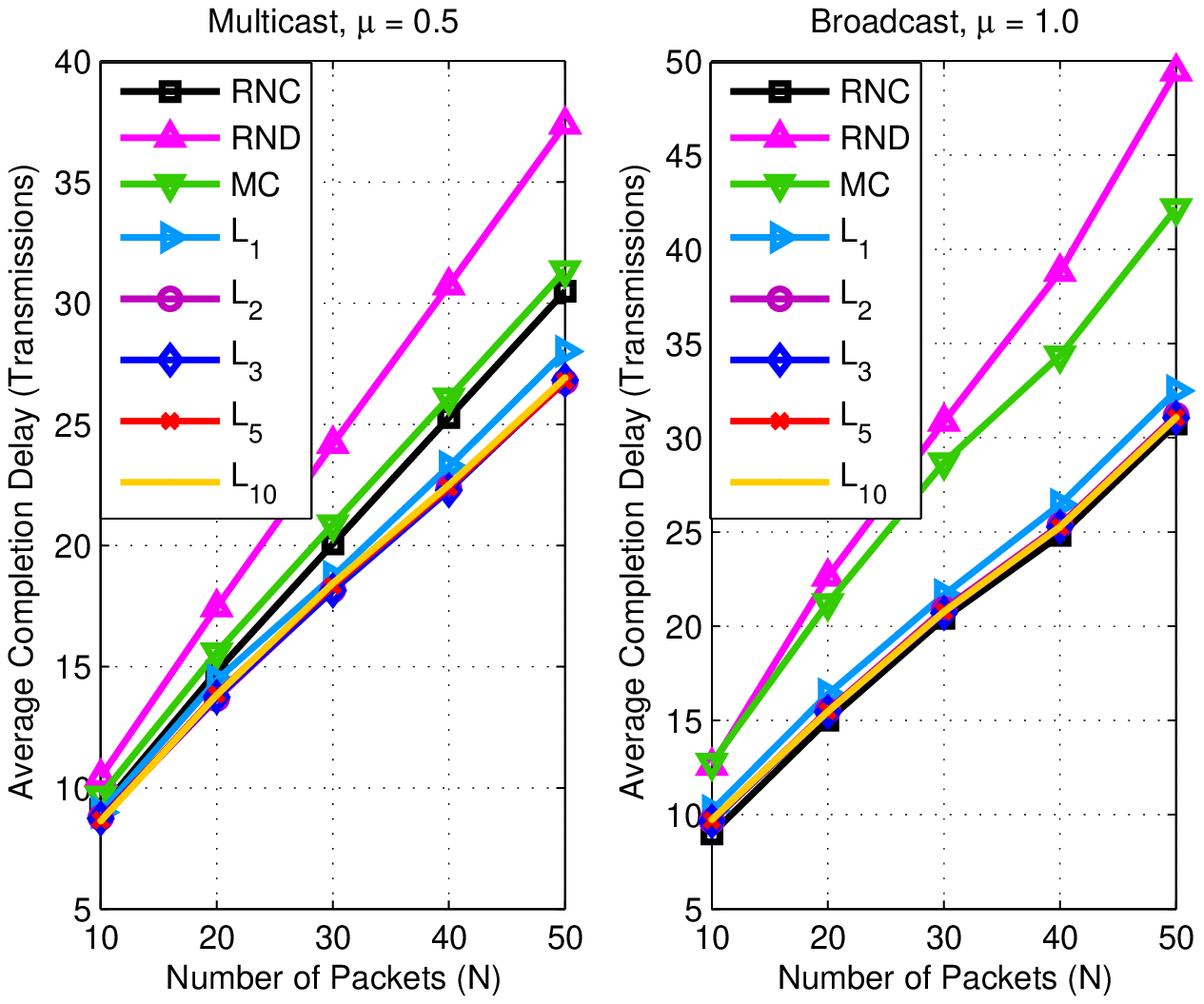}\\
  \caption{Performance comparison of different schemes vs $N$}\label{fig:N-opt}
\end{figure}

\begin{figure}[t]
\centering
  \includegraphics[width=0.95\linewidth]{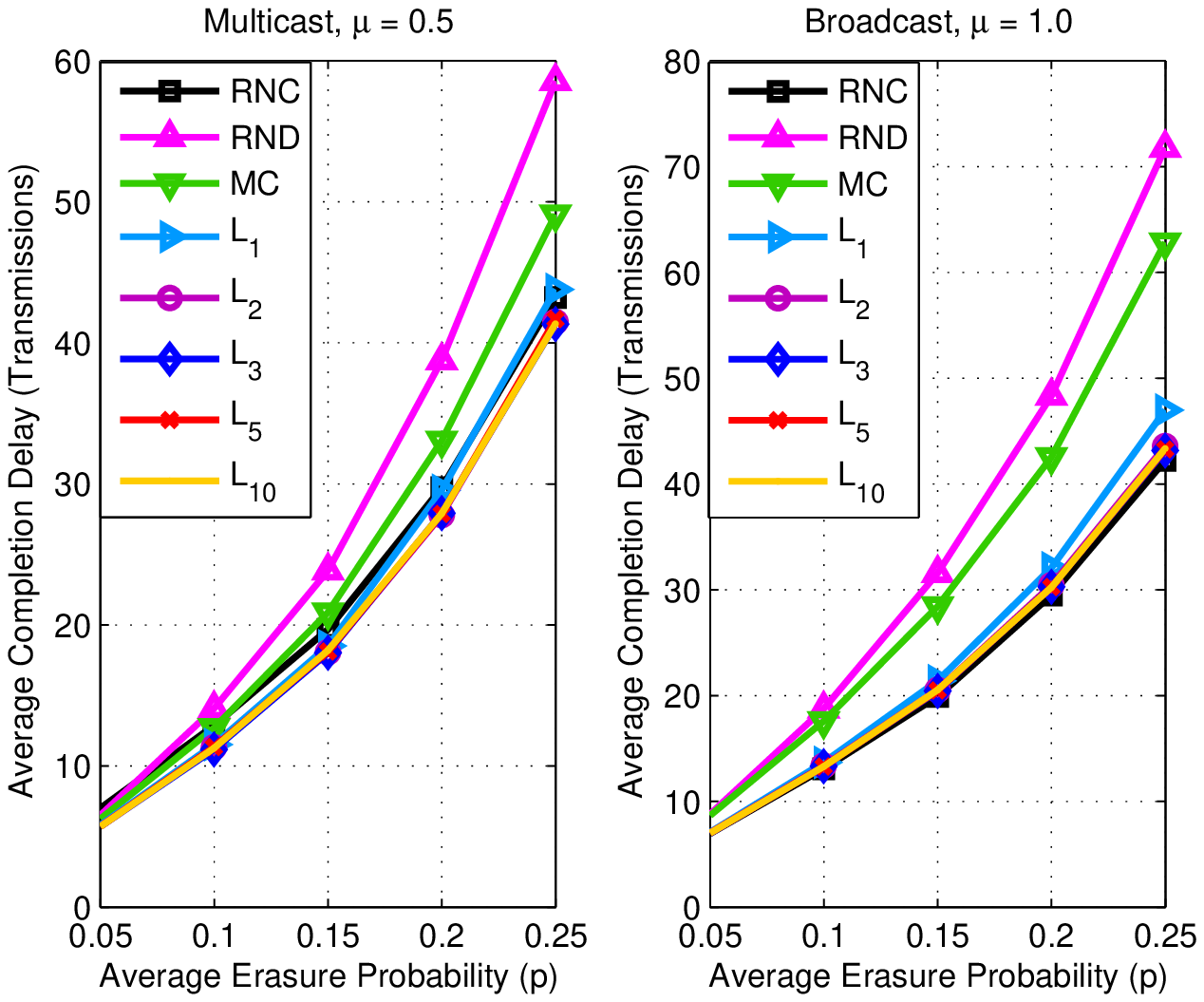}\\
  \caption{Performance comparison of different schemes vs $p$}\label{fig:P-opt}
\end{figure}

From all these figures, we can draw the following observations:
\begin{itemize}
\item Our proposed maximum weight clique selection algorithm with all norms considerably outperforms the RND and MC selection algorithms in terms of average completion delay for all comparison parameters ($\mu, M, N$ and $p$), Note that the MC is the IDNC version of the algorithms proposed in \cite{Wang2010,Gatzianas2010}, which shows the superior performance of our algorithms in the context of completion delay.
\item The $L_1$ norm algorithm employed in \cite{ICC10} degrades in performance, compared to the other norms, since it represents a very loose upper-bound of $V_{\pi^*}(s)$.
\item For norms higher than $1$, the algorithm tends to converge to the same performance with the smallest completion delays achieved by the $L_3$ and $L_5$ algorithms. For greater norms like $L_{10}$, the performance slightly degrades for all comparison parameters.
\item For the broadcast case ($\mu = 1$), results show that our proposed algorithm almost achieves the optimal performance of random network coding for all comparison parameters, with a maximum degradation of less than $5\%$ only occurring at very high number of receivers. This near-optimal performance is achieved while fully preserving the benefits of IDNC compared to perfect RNC.
\end{itemize}

Figures \ref{fig:D-srh}, \ref{fig:M-srh}, \ref{fig:N-srh} and \ref{fig:P-srh} compare the average completion delay of our proposed optimal maximum weight clique selection (denoted by ``opt") to that of our proposed maximum weight vertex search (denoted by ``srh'') algorithm for $L_3$, $L_5$ and $L_{10}$, as well as the maximum clique algorithm, using the same simulation parameters in Figures \ref{fig:D-opt} \ref{fig:M-opt}, \ref{fig:N-opt} and \ref{fig:P-opt}, respectively. For the MC approach, the heuristic algorithm is the same as the one described in \sref{sec:MWVS}, in which the $\left(\widetilde{\psi}_i\right)^n$ value of a vertex is replaced by its absolute primary degree.

\begin{figure}[t]
\centering
  \includegraphics[width=0.95\linewidth]{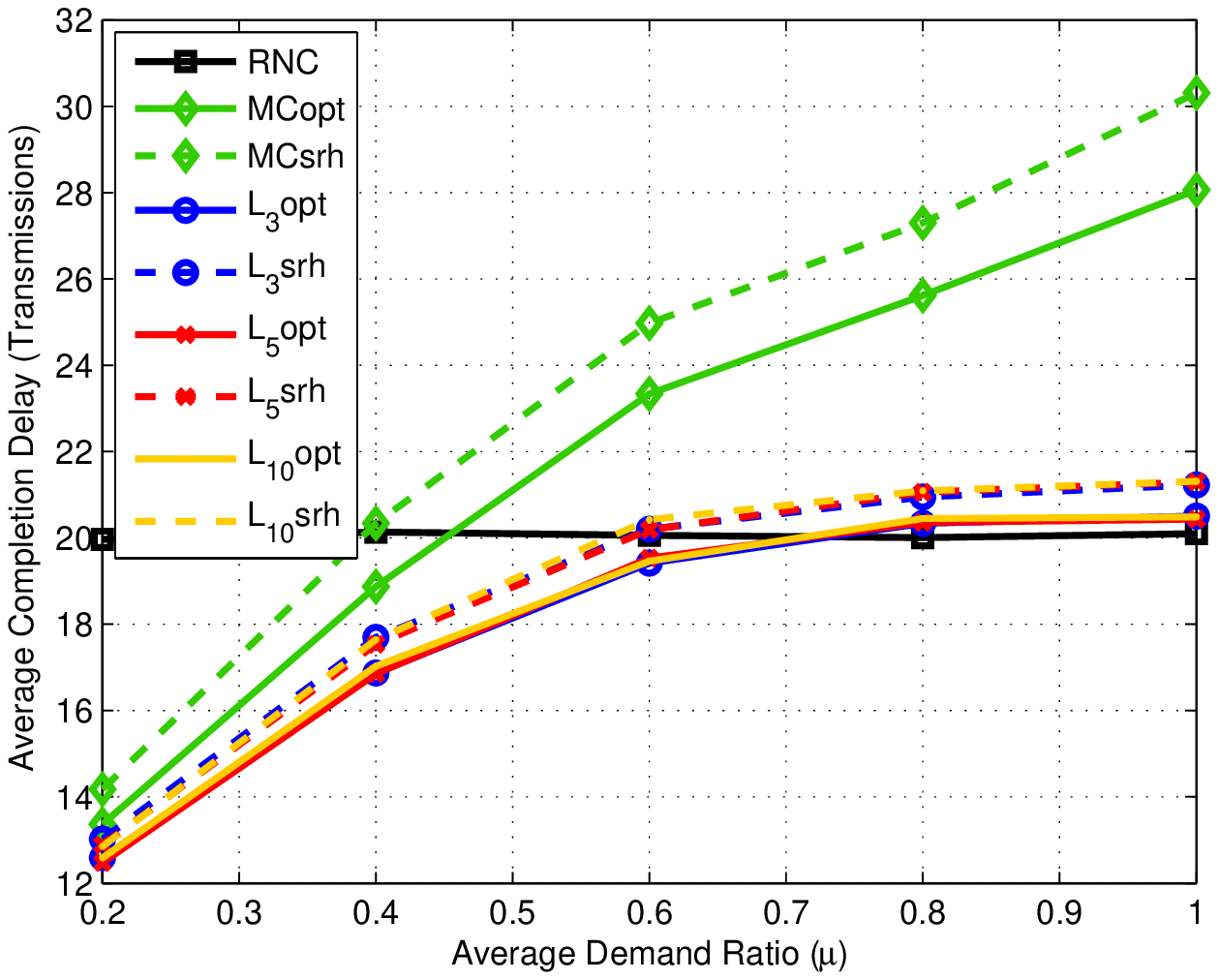}
  \caption{Performance comparison of optimal and heuristic clique search algorithms vs $\mu$}\label{fig:D-srh}
\end{figure}

\begin{figure}[t]
\centering
  \includegraphics[width=0.95\linewidth]{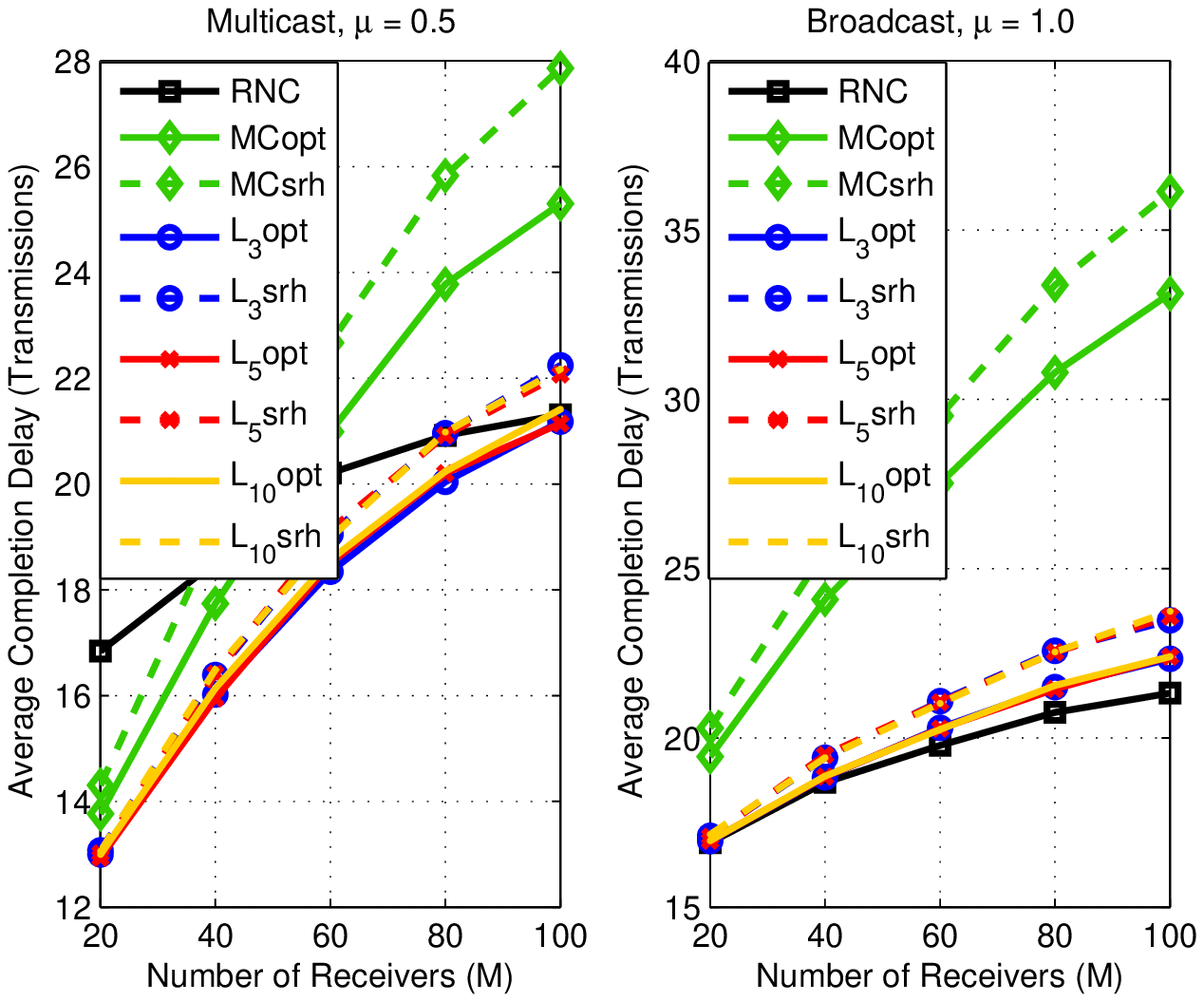}
  \caption{Performance comparison of optimal and heuristic clique search algorithms vs $M$}\label{fig:M-srh}
\end{figure}

\begin{figure}[t]
\centering
  \includegraphics[width=0.95\linewidth]{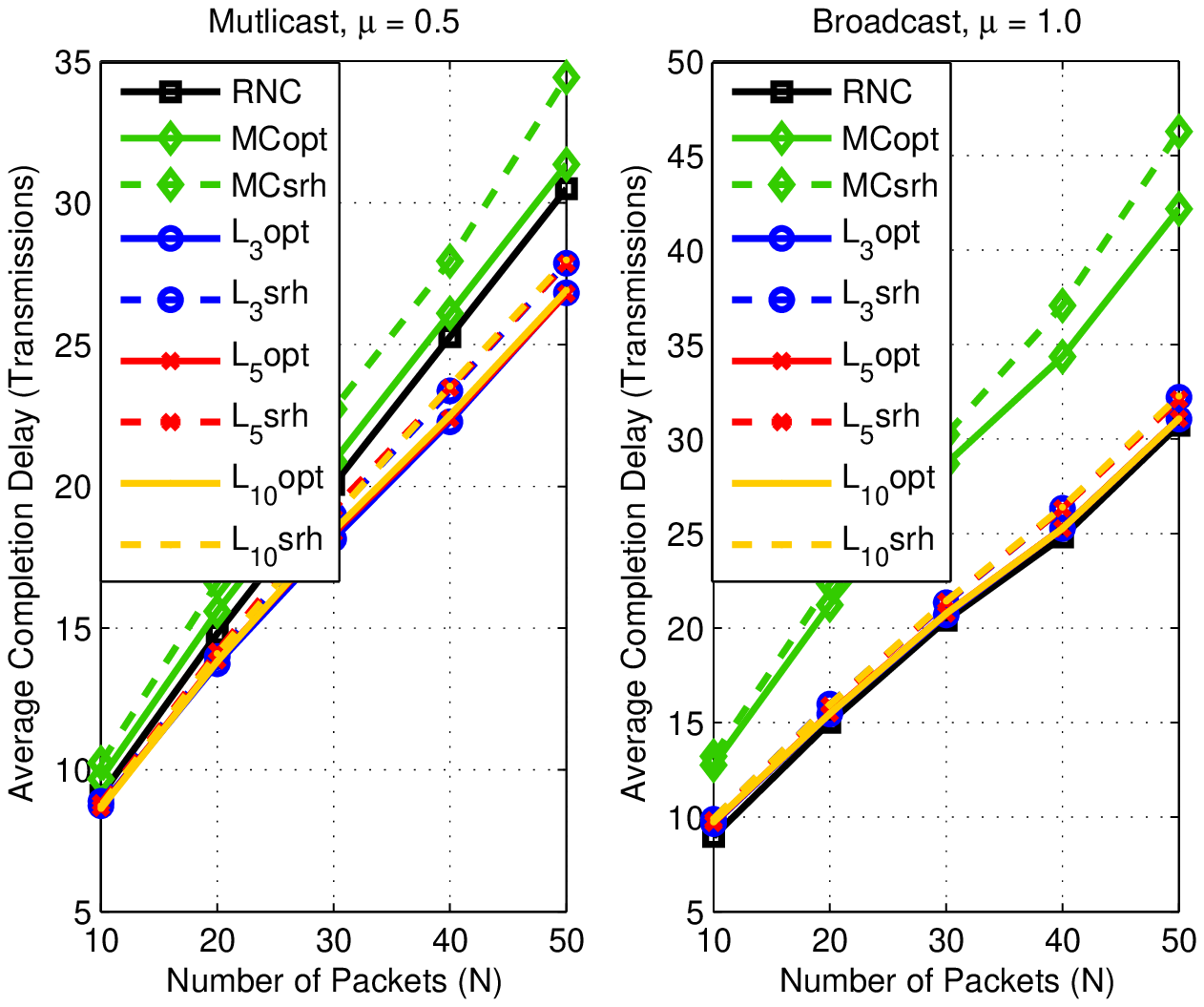}
  \caption{Performance comparison of optimal and heuristic clique search algorithms vs $N$}\label{fig:N-srh}
\end{figure}

\begin{figure}[t]
\centering
  \includegraphics[width=0.95\linewidth]{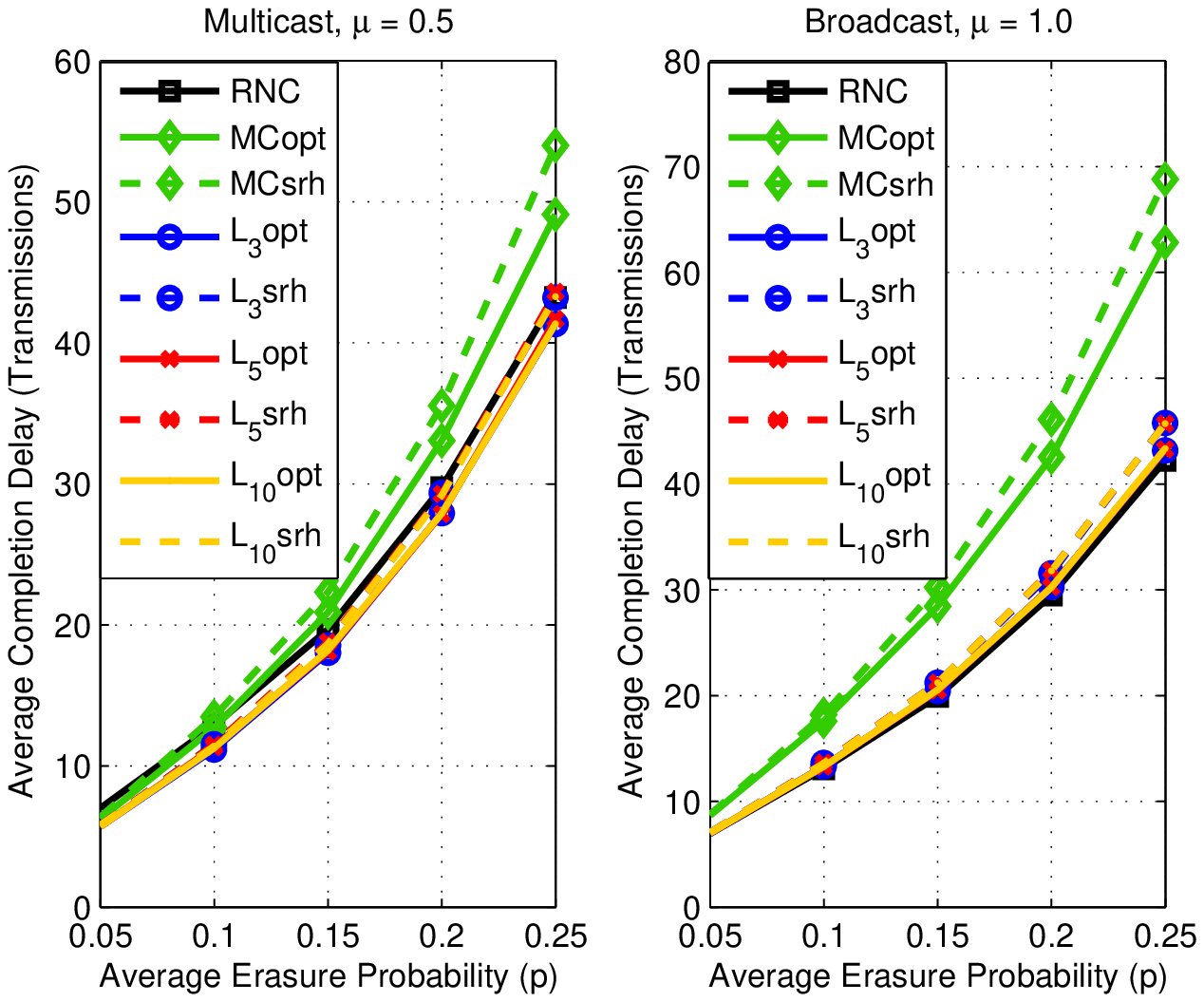}
  \caption{Performance comparison of optimal and heuristic clique search algorithms vs $p$}\label{fig:P-srh}
\end{figure}

For all these figures, we can see that the heuristic algorithms perform very closely to the optimal clique selection algorithms for all norms and all comparison parameters, with a maximum degradation of less than $5\%$ only occurring for large numbers of receivers. We can also observe a considerable improvement of our proposed heuristic algorithm with all norms compared to both the optimal and heuristic MC algorithms.

\section{Conclusion} \label{sec:conclusion}
In this paper, we studied the problem of minimizing the completion delay in wireless multicast and broadcast scenarios for IDNC.\ignore{ We first extended the IDNC graph representation to efficiently operate in both multicast and broadcast scenarios. We then} We formulated the minimum completion delay problem for IDNC as a stochastic shortest path (SSP) problem and showed that it is intractable. Nonetheless, we were able to draw the theoretical properties of the policies, which can efficiently reduce the IDNC completion delay, using the properties of the formulated SSP and the nature of the IDNC graph evolution. Based on these properties, we designed an efficient IDNC algorithm using a maximum weight clique selection algorithm, which can be solved in polynomial time. For further complexity reduction, we also designed a quadratic time heuristic algorithm, which can operate in real-time applications. Simulation results showed that our proposed heuristic can both efficiently track the optimal performance and outperform the random and maximum clique algorithms. For the broadcast case, simulations also showed that our proposed algorithms almost achieve the globally optimal performance of random network coding, while preserving all the benefits and simplicity of IDNC.
\ignore{
The only remaining challenge against making our IDNC algorithm a very efficient yet very simple solution to implement in real networks is its need for full and accurate feedback after each recovery transmission for IDNC graph update. In our future work, we aim to study the problem of optimizing the average completion delay for IDNC in infrequent and lossy feedback environments.
}

\appendices



\section{Proof of \thref{th:graph-density-evolution}} \label{app:graph-density-evolution}
\ignore{To prove this theorem, we first introduce the following lemma, which the multicast variant of Theorem 3 in \cite{TON10-CE}.}
To prove this theorem, we first introduce the following two theorems, proved in Appendices \ref{app:graph-density} and \ref{app:degree-evolution}, respectively.
\begin{theorem} \label{lem:graph-density}
For given $\boldsymbol{\varrho}$, $\boldsymbol{\varphi}$ and $\boldsymbol{\psi}$ vectors, the expected edge set cardinality of the primary graph is equal to:
\begin{equation}
\mathds{E}\left[\left|\mathcal{E}_\rho\right|\right] = \frac{1}{2}\sum_{i=1}^M \psi_i\left\{\sum_{\substack{k = 1\\k\neq i}}^M\: \frac{\psi_k}{N} \left(1+\frac{\varrho_k \varrho_i}{N-1}\right)\right\} \;. \label{eq:expected_edge_set_cardinality}
\end{equation}
\end{theorem}
\ignore{We also introduce the following theorem proved in \appref{app:degree-evolution}.}
\begin{theorem} \label{th:degree-evolution}
For a given maximal clique $\kappa$, chosen for transmission at time $t$, the expected primary degree of a vertex of receiver $i\in\mathcal{T}(\kappa)$ at time $t+1$ is expressed as:
\begin{equation}
\mathds{E}\left[\Delta_{i\in\mathcal{T}(\kappa)}^{(t+1)}\right]  =  \mathds{E}\left[\Delta_i^{(t)}\right] + \alpha_i \;. \ignore{\sum_{\substack{k=1\\k\neq i}}^M q_i \xi_k\ignore{\nonumber \\
&} -\sum_{\substack{k\in\mathcal{T}_\rho(\kappa)\\k\neq i}} \Phi_{ik}(q_i)\nonumber \\
& + \sum_{\substack{k\in\mathcal{T}_\sigma(\kappa)\\k\neq i}} \Lambda_{ik}(q_i) }\label{eq:degree-evolution-1}
\end{equation}
For $i\notin\mathcal{T}(\kappa)$, it is expressed as:
\begin{equation}
\mathds{E}\left[\Delta_{i\notin\mathcal{T}(\kappa)}^{(t+1)}\right]  = \mathds{E}\left[\Delta_i^{(t)}\right] + \beta_i \;. \ignore{-\sum_{\substack{k\in\mathcal{T}_\rho(\kappa)\\k\neq i}} \Phi_{ik}(0)\ignore{\nonumber \\
&} + \sum_{\substack{k\in\mathcal{T}_\sigma(\kappa)\\k\neq i}} \Lambda_{ik}(0)\;.} \label{eq:degree-evolution-2}
\end{equation}
\ignore{, such that:
\begin{align}
&\Phi_{ik}(x) = \frac{q_k}{N} \left( 1 + \frac{\left(\varrho_k-\psi_k+1\right) \left(\varrho_i+x\right)}{N-1}\right) \\
&\Lambda_{ik}(x)  = \frac{q_k \psi_k\left(\varrho_i+x\right)}{N(N-1)}\\
&\xi_k = \frac{\psi_k \varrho_k}{N(N-1)}\;.
\end{align}}
\end{theorem}
When the maximal clique $\kappa$ is chosen for transmission at time $t$, each member $k$ of the targeted receiver set $\mathcal{T}(\kappa)$ may (may not) receive the coded packet with probability $q_k$ ($p_k$).\ignore{When the maximal clique $\kappa$ is chosen for transmission at time $t$, the receivers in $\mathcal{T}(\kappa)$ are targeted with the coded packet but are not guaranteed to receive that packet.} Let $X_k$ be the random variable representing the reception of receiver $k\in\mathcal{T}(\kappa)$ at time $t$ and $\mathbf{X}$ as the random vector of all such random variables. From \thref{lem:graph-density}, we can derive the expression of the expected edge set size at time $t+1$, conditioned on the random vector $\mathbf{X}$, as follows:
\ignore{
\begin{align}
&\mathds{E}\left[\mathcal{E}_\rho^{(t+1)}\right] \nonumber \\
& = \frac{1}{2} \sum_{i\in\mathcal{T}(\kappa)} (\psi_i - X_i) \times \nonumber\\
&\qquad\; \left\{ \sum_{\substack{k\in\mathcal{T}(\kappa)_\rho\\k\neq i}} \frac{\psi_k - X_k}{N}\left[1 + (\varrho_k+X_k)\left(1-\frac{\varphi_i - X_i}{N}\right)\right] \right. \nonumber\\
& \qquad\quad+ \sum_{\substack{k\in\mathcal{T}_\sigma(\kappa)\\k\neq i}} \frac{\psi_k}{N}\left[1 + (\varrho_k+X_k)\left(1-\frac{\varphi_i - X_i}{N}\right)\right] \nonumber\\
& \qquad\quad \left.\vphantom{\sum_{\substack{k\in\mathcal{T}_\rho(\kappa)\\k\neq i}} \frac{\psi_k - X_k}{N}\left[1 + (\varrho_k+X_k)\left(1-\frac{\varphi_i - X_i}{N}\right)\right]} + \sum_{k\notin\mathcal{T}(\kappa)} \frac{\psi_k}{N}\left[1 + (\varrho_k)\left(1-\frac{\varphi_i - X_i}{N}\right)\right] \right\}\nonumber
\end{align}
}

\begin{align}
&\mathds{E}\left[\left|\mathcal{E}_\rho^{(t+1)}\right|\Big|\mathbf{X}\right] =  \frac{1}{2} \sum_{i\in\mathcal{T}_\rho(\kappa)} (\psi_i - X_i)\; \mathds{E}\left[\Delta_{i\in\mathcal{T}(\kappa)}^{(t+1)}\Big|\mathbf{X}\right] \nonumber\\
& + \frac{1}{2} \sum_{i\in\mathcal{T}_\sigma(\kappa)} \psi_i \; \mathds{E}\left[\Delta_{i\in\mathcal{T}(\kappa)}^{(t+1)}\Big|\mathbf{X}\right] + \frac{1}{2} \sum_{i\notin\mathcal{T}(\kappa)} \psi_i \; \mathds{E}\left[\Delta_{i\notin\mathcal{T}(\kappa)}^{(t+1)}\Big|\mathbf{X}\right] \nonumber\\
& =  \frac{1}{2} \sum_{i\in\mathcal{T}(\kappa)} \psi_i \;\mathds{E}\left[\Delta_{i\in\mathcal{T}(\kappa)}^{(t+1)}\Big|\mathbf{X}\right] + \frac{1}{2} \sum_{i\notin\mathcal{T}(\kappa)} \psi_i \; \mathds{E}\left[\Delta_{i\notin\mathcal{T}(\kappa)}^{(t+1)}\Big|\mathbf{X}\right] \nonumber\\
&- \frac{1}{2}\sum_{i\in\mathcal{T}_\rho(\kappa)} X_i \; \mathds{E}\left[\Delta_{i\in\mathcal{T}(\kappa)}^{(t+1)}\Big|\mathbf{X}\right] \nonumber \;.
\end{align}
Now, taking the expectation operator over the random vector $\mathbf{X}$, and using the expressions of \thref{th:degree-evolution}, we get the expression of the expected edge set size at time $t+1$ as follows:
\begin{align}
\mathds{E}\left[\left|\mathcal{E}_\rho^{(t+1)}\right|\right] & =  \mathds{E}_{\mathbf{X}} \left\{ \mathds{E}\left[\mathcal{E}_\rho^{(t+1)}\Big|\mathbf{X}\right]\right\} \nonumber \\
\ignore{& =  \frac{1}{2} \sum_{i\in\mathcal{T}(\kappa)} \psi_i\; \mathds{E}_{\mathbf{X}}\left\{\mathds{E}\left[\Delta_{i\in\mathcal{T}(\kappa)}^{(t+1)}\Big|\mathbf{X}\right]\right\}\ignore{ \nonumber\\
&}+ \frac{1}{2} \sum_{i\notin\mathcal{T}(\kappa)} \psi_i \; \mathds{E}_{\mathbf{X}}\left\{\mathds{E}\left[\Delta_{i\notin\mathcal{T}(\kappa)}^{(t+1)}\Big|\mathbf{X}\right]\right\} \nonumber\\
& - \sum_{i\in\mathcal{T}_\rho(\kappa)} \mathds{E}_{\mathbf{X}}\left\{X_i \; \mathds{E}\left[\Delta_{i\in\mathcal{T}(\kappa)}^{(t+1)}\Big|\mathbf{X}\right]\right\} \nonumber\\}
& =   \frac{1}{2} \sum_{i\in\mathcal{T}(\kappa)} \psi_i\; \left\{\mathds{E}\left[\Delta_i^{(t)}\right] + \alpha_i\right\} \nonumber \\
& + \frac{1}{2} \sum_{i\notin\mathcal{T}(\kappa)} \psi_i \; \left\{\mathds{E}\left[\Delta_i^{(t)}\right] + \beta_i\right\} \nonumber\\
& - \sum_{i\in\mathcal{T}_\rho(\kappa)} \mathds{E}_{\mathbf{X}}\left\{X_i \; \mathds{E}\left[\Delta_{i\in\mathcal{T}(\kappa)}^{(t+1)}\Big|\mathbf{X}\right]\right\}
\end{align}
Grouping similar terms and expanding $\mathds{E}\left[\Delta_{i\in\mathcal{T}(\kappa)}^{(t+1)}\Big|\mathbf{X}\right]$, we get:

\begin{align}
& \mathds{E}\left[\left|\mathcal{E}_\rho^{(t+1)}\right|\right] = \mathds{E}\left[\left|\mathcal{E}_\rho^{(t)}\right|\right] +  \frac{1}{2} \sum_{i\in\mathcal{T}(\kappa)} \psi_i \alpha_i + \frac{1}{2} \sum_{i\notin\mathcal{T}(\kappa)} \psi_i \beta_i \nonumber\\
& - \mathds{E}_{\mathbf{X}}\left\{X_i\right\}\mathds{E}\left[\Delta_i^{(t)}\right] + \sum_{\substack{k=1\\k\neq i}}^M  \frac{\psi_k\varrho_k \mathds{E}_{\mathbf{X}}\left\{X_i^2\right\}}{N(N-1)} \nonumber\\
\ignore{
&  - \sum_{\substack{k\in\mathcal{T}_\rho(\kappa)\\k\neq i}} \frac{\mathds{E}_{\mathbf{X}}\{X_i X_k\}}{N} + \left[\frac{\mathds{E}_{\mathbf{X}}\{ X_k\}}{N}\left(\varrho_k - \psi_k\right) + \frac{\mathds{E}_{\mathbf{X}}\{ X_k^2\}}{N}\right]\nonumber \\
& \qquad\qquad\qquad\qquad\qquad\qquad\qquad\qquad\times\;\mathds{E}_{\mathbf{X}}\left\{\frac{\varrho_i X_i + X^2_i}{N-1}\right\} \nonumber\\
}
&  - \sum_{\substack{k\in\mathcal{T}_\rho(\kappa)\\k\neq i}} \mathds{E}_{\mathbf{X}}\left\{\frac{X_i X_k}{N} + \frac{\left(X_k\left(\varrho_k - \psi_k\right) + X_k^2\right)\left(\varrho_i X_i + X^2_i\right)}{N(N-1)}\right\} \nonumber\\
& + \sum_{\substack{k\in\mathcal{T}_\sigma(\kappa)\\k\neq i}} \mathds{E}_{\mathbf{X}}\left\{\frac{\psi_k X_k\left(\varrho_i X_i + X^2_i\right)}{N(N-1)}\right\} \nonumber\\
& = \mathds{E}\left[\left|\mathcal{E}_\rho^{(t)}\right|\right]- \frac{1}{2} \sum_{i\in\mathcal{T}_\rho(\kappa)} q_i \left(\mathds{E}\left[\Delta_i^{(t)}\right] + \gamma_i\right)  \nonumber\\
& + \frac{1}{2}\sum_{i\in\mathcal{T}(\kappa)}  \psi_i \alpha_i + \frac{1}{2}\sum_{i\notin\mathcal{T}(\kappa)}  \psi_i \beta_i\;.
\end{align}

\section{Proof of \thref{th:degree-comparison}} \label{app:degree-comparison}
To prove this theorem, we first introduce the following theorem, proved in \appref{app:expected-degree}\ignore{, which is the multicast variant of Lemma 1 in \cite{TON10-CE}}.
\begin{theorem} \label{lem:expected-degree}
For given $\boldsymbol{\varrho}$, $\boldsymbol{\varphi}$ and $\boldsymbol{\psi}$ vectors, the expected primary degree of any of the vertices induced by receiver $i$ (denoted by $\Delta_{i}$) is equal to:
\begin{equation}
\mathds{E}\left[\Delta_{i}\right] = \sum_{\substack{k = 1\\k\neq i}}^M\: \frac{\psi_k}{N} \left(1+\frac{\varrho_k\varrho_i}{N-1}\right)\;. \label{eq:primary_expected_degree}
\end{equation}
\end{theorem}
Expanding the expression in \eqref{eq:primary_expected_degree}, we get:
\begin{align} \label{eq:degree-comparison}
\mathds{E}\left[\Delta_h\right] =& \sum_{\substack{k = 1\\k\neq i,h}}^M\: \frac{\psi_k}{N} \left(1+\frac{\varrho_k\varrho_h}{N-1}\right) + \frac{\psi_i}{N} \left(1+\frac{\varrho_i\varrho_h}{N-1}\right) \nonumber \\
> &\sum_{\substack{k = 1\\k\neq i,h}}^M\: \frac{\psi_k}{N} \left(1+\frac{\varrho_k\varrho_i}{N-1}\right) \ignore{\nonumber \\
&}+ \frac{\psi_h}{N} \left(1+\frac{\varrho_i\varrho_h}{N-1}\right) \nonumber \\
 =& \mathds{E}\left[\Delta_i\right]\;.
\end{align}

\section{Proof of \thref{th:alpha-beta}} \label{app:alpha-beta}
\begin{align}
\alpha_i = &  \sum_{\substack{k=1\\k\neq i}}^M  \frac{q_i\psi_k \varrho_k}{N(N-1)} + \sum_{\substack{k\in\mathcal{T}_\sigma(\kappa)\\k\neq i}} \frac{q_k \psi_k\left(\varrho_i+q_i\right)}{N(N-1)}\nonumber\\
& -\sum_{\substack{k\in\mathcal{T}_\rho(\kappa)\\k\neq i}}\frac{q_k}{N} \left( 1 + \frac{\left(\varrho_k-\psi_k+1\right) \left(\varrho_i+q_i\right)}{N-1}\right)
\end{align}
Re-arranging the above expression and using the definition of $\beta_i$ in \eqref{eq:beta}, we get:
\begin{align} \label{eq:alpha-beta-proof-I}
\alpha_i  =& \; \beta_i + \sum_{\substack{k\notin\mathcal{T}(\kappa)\\k\neq i}}  \frac{q_i\psi_k \varrho_k}{N(N-1)} + \sum_{\substack{k\in\mathcal{T}_\sigma(\kappa)\\k\neq i}}  \frac{q_i\psi_k \left(\varrho_k+q_i\right)}{N(N-1)} \nonumber\\
& +\sum_{\substack{k\in\mathcal{T}_\rho(\kappa)\\k\neq i}}\frac{q_i\psi_k\varrho_k - q_iq_k\left(\varrho_k-\psi_k+1\right)}{N(N-1)}
\end{align}
Since for $\psi_k>0$, $\varrho_k \geq \varrho_k - \psi_k +1$ and $\psi_k\ > q_k$, the last term in \eqref{eq:alpha-beta-proof-I} is non-negative and the theorem follows.

\section{Proof of \thref{lem:graph-density}} \label{app:graph-density}
It is well known from graph theory that the edge set size of any graph is equal to half the sum of its vertex degrees. Consequently, we will find an expression for the sum of the vertex primary degrees as follows. From \appref{app:expected-degree}, we know that the primary degree of a vertex $v_{ij}$ is expressed as in \eqref{eq:vertex-degree}. Consequently, the sum $\Sigma \Delta_i$ of all the primary degrees of the vertices induced by receiver $i$ can be expressed as:
\begin{align}
&\Sigma \Delta_i = \sum_{j\in\mathcal{W}_i} \sum_{\substack{k=1\\k\neq i}}^M \bigg[I_{j\in\mathcal{W}_k} + I_{j\in\mathcal{R}_k}\left(\left|\mathcal{W}_k\right| - \left|\mathcal{W}_k \cap \mathcal{L}_i\right|\right)\bigg] \nonumber \\
& = \sum_{\substack{k=1\\k\neq i}}^M \left[\left(\sum_{j\in\mathcal{W}_i}I_{j\in\mathcal{W}_k}\right) + \sum_{j\in\mathcal{W}_i}I_{j\in\mathcal{R}_k}\cdot \left(\psi_k - \left|\mathcal{W}_k \cap \mathcal{L}_i\right|\right)\right] \nonumber\\
& = \sum_{\substack{k=1\\k\neq i}}^M \left[\left|\mathcal{W}_i \cap \mathcal{W}_k\right|  + \sum_{j\in\mathcal{W}_i}I_{j\in\mathcal{R}_k}\cdot \left(\psi_k - \left|\mathcal{W}_k \cap \mathcal{L}_i\right|\right) \right]\;.
\end{align}
Now, ignoring the content of the different sets, we can derive an expression for the expected edge set size of the primary graph. Note that\ignore{ when the content of different sets are ignored,} the cardinality of the intersection of any two sets of given sizes, whose elements are unknown but are drawn from the same pool of elements, is a hypergeometric distributed random variable.\ignore{ $\left|\mathcal{W}_i \cap \mathcal{R}_k\right|$ and $\left|\mathcal{W}_k \cap \mathcal{L}_i\right|$ become independent hypergeometric random variables.} Thus, we get:
\begin{align}\label{eq:lem-II-proof-II}
&\mathds{E}\left[\left|\mathcal{E}_\rho\right|\right] = \frac{1}{2}\sum_{i=1}^M \mathds{E}\left[\Sigma \Delta_i\right] \nonumber\\
& = \frac{1}{2}\sum_{i=1}^M \sum_{\substack{k=1\\k\neq i}}^M \left\{ \mathds{E}\left[\left|\mathcal{W}_i \cap \mathcal{W}_k\right|\right]  + \psi_k \mathds{E}\left[\sum_{j\in\mathcal{W}_i}I_{j\in\mathcal{R}_k}\right]\right. \nonumber \\
& \left. \qquad\qquad\qquad - \mathds{E}\left[\sum_{j\in\mathcal{W}_i}I_{j\in\mathcal{R}_k}\cdot\left|\mathcal{W}_k \cap \mathcal{L}_i\right|\right]\right\} \nonumber\\
& = \frac{1}{2} \sum_{i=1}^M \sum_{\substack{k=1\\k\neq i}}^M \left\{\frac{\psi_i\psi_k}{N}  + \frac{\psi_k \psi_i\varrho_k}{N} - \sum_{j\in\mathcal{W}_i} \mathds{E}\left[I_{j\in\mathcal{R}_k} \cdot \left|\mathcal{W}_k \cap \mathcal{L}_i\right|\right]\right\}
\end{align}
Note that the indicator function in the last term expectation can be only zero or one. Consequently, the expectation of its multiplication with $\left|\mathcal{W}_k \cap \mathcal{L}_i\right|$ can be only evaluated for $I_{j\in\mathcal{R}_k} = 1$. If this is the case, this means that packet $j$ cannot be in the intersection of $\mathcal{W}_k$ and $\mathcal{L}_i$. Consequently, this intersection is possible only with the other $\varphi_i-1$ packets and from the set of the remaining $N-1$ packets. Thus, we get:\ignore{ Since the cardinality of the intersection of two sets of given sizes, whose elements are drawn of the same pool of $N-1$ elements, is a hypergeometric random variable, we have:}
\begin{align} \label{eq:lem-I-proof-II}
& \mathds{E}\left[I_{j\in\mathcal{R}_k} \cdot \left|\mathcal{W}_k \cap \mathcal{L}_i\right|\right] = \sum_{n = 1}^{N-1}
n \: \mathds{P}\left[I_{j\in\mathcal{R}_k} = 1, \left|\mathcal{W}_k \cap \mathcal{L}_i\right| = n \right] \nonumber\\
&  = \sum_{n = 1}^{N-1}
n \: \mathds{P}\left[\left|\mathcal{W}_k \cap \mathcal{L}_i\right| = n \Big| I_{j\in\mathcal{R}_k} = 1\right] \cdot \mathds{P}\left[I_{j\in\mathcal{R}_k} = 1\right] \nonumber \\
& = \sum_{n=1}^{N-1} n \frac{\binom{\varphi_i-1}{n}\binom{N-1-\varphi_i+1}{\psi_k - n}}{\binom{N-1}{\psi_k}} \:\frac{\varrho_k}{N} = \frac{\varrho_k \psi_k \left(\varphi_i - 1\right)}{N(N-1)}\;.
\end{align}

Substituting \eqref{eq:lem-I-proof-II} in \eqref{eq:lem-II-proof-II}, we get:
\begin{align}
&\mathds{E}\left[\left|\mathcal{E}_\rho\right|\right] = \frac{1}{2} \sum_{i=1}^M \sum_{\substack{k=1\\k\neq i}}^M \frac{\psi_i\psi_k}{N}  + \frac{\psi_k \psi_i\varrho_k}{N} - \frac{\psi_i \varrho_k \psi_k(\varphi_i-1)}{N(N-1)} \nonumber\\
& = \frac{1}{2} \sum_{i=1}^M \psi_i \left\{ \sum_{\substack{k=1\\k\neq i}}^M \frac{\psi_k}{N} \left[ 1 + \varrho_k \left(1 - \frac{\varphi_i-1}{N-1}\right)\right]\right\} \nonumber\\
& = \frac{1}{2} \sum_{i=1}^M \psi_i \left\{\sum_{\substack{k=1\\k\neq i}}^M \frac{\psi_k}{N} \left(1 + \frac{\varrho_k \varrho_i}{N-1}\right)\right\}\;.
\end{align}

\section{Proof of \thref{th:degree-evolution}} \label{app:degree-evolution}
\ignore{To prove this theorem, we first introduce the following theorem, proved in \appref{app:expected-degree}\ignore{, which is the multicast variant of Lemma 1 in \cite{TON10-CE}}.
\begin{theorem} \label{lem:expected-degree}
For given $\boldsymbol{\varrho}$, $\boldsymbol{\varphi}$ and $\boldsymbol{\psi}$ vectors, the expected primary degree of any of the vertices induced by receiver $i$ (denoted by $\Delta_{i}$) is equal to:
\begin{equation}
\mathds{E}\left[\Delta_{i}\right] = \sum_{\substack{k = 1\\k\neq i}}^M\: \frac{\psi_k}{N} \left(1+\frac{\varrho_k\varrho_i}{N-1}\right)\;. \label{eq:primary_expected_degree}
\end{equation}
\end{theorem}
When the maximal clique $\kappa$ is chosen for transmission at time $t$, each member $k$ of the targeted receiver set $\mathcal{T}(\kappa)$ may (may not) receive the coded packet with probability $q_k$ ($p_k$). Let $X_k$ be the random variable representing the reception of receiver $k\in\mathcal{T}(\kappa)$ at time $t$ and $\mathbf{X}$ as the random vector of all such random variables.}Using \eqref{eq:primary_expected_degree} in \thref{lem:expected-degree}, we can derive an expression for the expected primary degree of receiver $i\in\mathcal{T}(\kappa)$ at time $t+1$, conditioned on the random vector $\mathbf{X}$ defined in \appref{app:graph-density-evolution}, as follows:
\begin{align}
&\mathds{E}\left[\Delta_{i\in\mathcal{T}(\kappa)}^{(t+1)}\Big|\mathbf{X}\right] \nonumber\\
&= \sum_{\substack{k\in\mathcal{T}_\rho(\kappa)\\k\neq i}} \frac{\psi_k - X_k}{N}\left(1 + \frac{\left(\varrho_k+X_k\right)\left(\varrho_i + X_i\right)}{N-1}\right) \nonumber\\
& \qquad\quad+ \sum_{\substack{k\in\mathcal{T}_\sigma(\kappa)\\k\neq i}} \frac{\psi_k}{N}\left(1 + \frac{\left(\varrho_k+X_k\right)\left(\varrho_i + X_i\right)}{N-1}\right) \nonumber\\
& \qquad\quad + \sum_{k\notin\mathcal{T}(\kappa)} \frac{\psi_k}{N}\left(1 + \frac{\varrho_k\left(\varrho_i + X_i\right)}{N-1}\right)
\end{align}
Re-arranging the terms, we get:
\begin{align}
&\mathds{E}\left[\Delta_{i\in\mathcal{T}(\kappa)}^{(t+1)}\Big|\mathbf{X}\right]
= \sum_{\substack{k=1\\k\neq i}}^M \frac{\psi_k}{N}\left(1 + \frac{\varrho_k \varrho_i}{N-1}\right) + \sum_{\substack{k=1\\k\neq i}}^M  \frac{\psi_k\varrho_k X_i}{N(N-1)}  \label{eq:reference}\\
& \qquad\quad- \sum_{\substack{k\in\mathcal{T}_\rho(\kappa)\\k\neq i}} \frac{X_k}{N}\left(1 + \frac{\left(\varrho_k -\psi_k + X_k\right)\left(\varrho_i + X_i\right)}{N-1}\right) \nonumber\\
& \qquad\quad+ \sum_{\substack{k\in\mathcal{T}_\sigma(\kappa)\\k\neq i}} \frac{\psi_k X_k \left(\varrho_i + X_i\right)}{N(N-1)}\;.
\end{align}
The first term in \eqref{eq:reference} is obviously the expected vertex degree of receiver $i$ at time $t$. Now, we can derive the expected degree of receiver $i$ after serving the maximal clique $\kappa$ as follows:
\begin{align}
&\mathds{E}\left[\Delta_{i\in\mathcal{T}(\kappa)}^{(t+1)}\right]  = \mathds{E}_{\mathbf{X}}\left\{\mathds{E}\left[\Delta_{i\in\mathcal{T}(\kappa)}^{(t+1)}\Big|\mathbf{X}\right]\right\} \nonumber\\
& = \mathds{E}\left[\Delta_i^{(t)}\right] + \sum_{\substack{k=1\\k\neq i}}^M  \frac{\psi_k\varrho_k \mathds{E}_{\mathbf{X}}\left\{X_i\right\}}{N(N-1)} \nonumber\\
&  - \sum_{\substack{k\in\mathcal{T}_\rho(\kappa)\\k\neq i}} \mathds{E}_{\mathbf{X}}\left\{\frac{X_k}{N} + \frac{\left(X_k\left(\varrho_k - \psi_k\right) + X_k^2\right)\left(\varrho_i + X_i\right)}{N(N-1)}\right\} \nonumber\\
& + \sum_{\substack{k\in\mathcal{T}_\sigma(\kappa)\\k\neq i}} \mathds{E}_{\mathbf{X}}\left\{\frac{\psi_k X_k\left(\varrho_i + X_i\right)}{N(N-1)}\right\} \nonumber\\
& = \mathds{E}\left[\Delta_i^{(t)}\right] + \sum_{\substack{k=1\\k\neq i}}^M q_i \xi_k\ignore{\nonumber \\
&} -\sum_{\substack{k\in\mathcal{T}_\rho(\kappa)\\k\neq i}} \Phi_{ik}(q_i)\ignore{\nonumber \\
&} + \sum_{\substack{k\in\mathcal{T}_\sigma(\kappa)\\k\neq i}} \Lambda_{ik}(q_i)\;.
\end{align}
The expression for $\mathds{E}\left[\Delta_{i\notin\mathcal{T}(\kappa)}^{(t+1)}\right]$  can be derived using the same approach.

\section{Proof of \thref{lem:expected-degree}} \label{app:expected-degree}
Consider an arbitrary vertex $v_{ij}$ in the graph. From the adjacency conditions C1 and C2 in \sref{sec:GIDNC-graph}, we can conclude the following facts:
\begin{itemize}
\item Vertex $v_{ij}$ is not connected to any vertex of the same receive $i$.
\item If $j\in\mathcal{W}_k$, $v_{ij}$ cannot be adjacent to any primary vertex of receiver $k$ due to violation of C2, except for vertex $v_{kj}$ which arises from C1.
\item If $j\in\mathcal{R}_k$, $v_{ij}$ can be connected to any primary vertex of receiver $k$ (induced from $\mathcal{W}_k$), except for all vertices $v_{kl}$ for which $l\notin\mathcal{R}_i \quad l\in\mathcal{L}_i$.
\end{itemize}
From these facts, we can express the primary degree of a vertex $v_{ij}$ as follows:
\begin{equation} \label{eq:vertex-degree}
\Delta_{ij} = \sum_{\substack{k=1\\k\neq i}}^M \bigg[I_{j\in\mathcal{W}_k} + I_{j\in\mathcal{R}_k}\left(\left|\mathcal{W}_k\right| - \left|\mathcal{W}_k \cap \mathcal{L}_i\right|\right)\bigg]\;
\end{equation}
where $I_x$ is an indicator function, which is equal to one if $x$ is true and zero otherwise.

Now, ignoring the content of the different sets, we can derive the expression for the expected primary degree of a vertex of receiver $i$. Consequently, we get:\
\begin{align} \label{eq:lem-I-proof-I}
& \mathds{E}\left[\Delta_i\right] = \mathds{E}\left[\Delta_{ij}\right] \nonumber\\
& = \sum_{\substack{k=1\\k\neq i}}^M \bigg[\mathds{E}\left[I_{j\in\mathcal{W}_k}\right] + \mathds{E}\left[I_{j\in\mathcal{R}_k}\right]\left|\mathcal{W}_k\right| - \mathds{E}\left[I_{j\in\mathcal{R}_k} \cdot \left|\mathcal{W}_k \cap \mathcal{L}_i\right|\right] \bigg]\nonumber\\
&= \sum_{\substack{k=1\\k\neq i}}^M \left[\frac{\psi_k}{N} + \frac{\varrho_k \psi_k}{N} - \mathds{E}\left[I_{j\in\mathcal{R}_k} \cdot \left|\mathcal{W}_k \cap \mathcal{L}_i\right|\right] \right]\;.
\end{align}
\ignore{Note that the indicator function in the last term expectation can be only zero or one. Consequently, the expectation of its multiplication with $\left|\mathcal{W}_k \cap \mathcal{L}_i\right|$ can be only evaluated for $I_{j\in\mathcal{R}_k} = 1$. If this is the case, this means that packet $j$ cannot be in the intersection of $\mathcal{W}_k$ and $\mathcal{L}_i$. Consequently, this intersection is possible only with the other $\varphi_i-1$ packets and from the set of the remaining $N-1$ packets. Since the cardinality of the intersection of two sets of given sizes, whose elements are drawn of the same pool of $N-1$ elements, is a hypergeometric random variable, we have:
\begin{align} \label{eq:lem-I-proof-II}
& \mathds{E}\left[I_{j\in\mathcal{R}_k} \cdot \left|\mathcal{W}_k \cap \mathcal{L}_i\right|\right] = \sum_{n = 1}^{N-1}
n \: \mathds{P}\left[I_{j\in\mathcal{R}_k} = 1, \left|\mathcal{W}_k \cap \mathcal{L}_i\right| = n \right] \nonumber\\
&  = \sum_{n = 1}^{N-1}
n \: \mathds{P}\left[\left|\mathcal{W}_k \cap \mathcal{L}_i\right| = n \Big| I_{j\in\mathcal{R}_k} = 1\right] \cdot \mathds{P}\left[I_{j\in\mathcal{R}_k} = 1\right] \nonumber \\
& = \sum_{n=1}^{N-1} n \frac{\binom{\varphi_i-1}{n}\binom{N-1-\varphi_i+1}{\psi_k - n}}{\binom{N-1}{\psi_k}} \:\frac{\varrho_k}{N} \nonumber \\
& = \frac{\varrho_k \psi_k \left(\varphi_i - 1\right)}{N(N-1)}\;.
\end{align}}
Substituting \eqref{eq:lem-I-proof-II} in \eqref{eq:lem-I-proof-I} and re-arranging, we get:
\begin{align}
\mathds{E}\left[\Delta_i\right] &= \sum_{\substack{k=1\\k\neq i}}^M \frac{\psi_k}{N} \left[ 1 + \varrho_k \left(1 - \frac{\varphi_i-1}{N-1}\right)\right] \nonumber \\
&= \sum_{\substack{k=1\\k\neq i}}^M \frac{\psi_k}{N} \left( 1 + \frac{\varrho_k \varrho_i}{N-1}\right)\;.
\end{align}

\bibliographystyle{IEEEtran}
\bibliography{IEEEabrv,bibfile}

\begin{IEEEbiography}[{\includegraphics[width=1in,height=1.25in,clip,keepaspectratio]{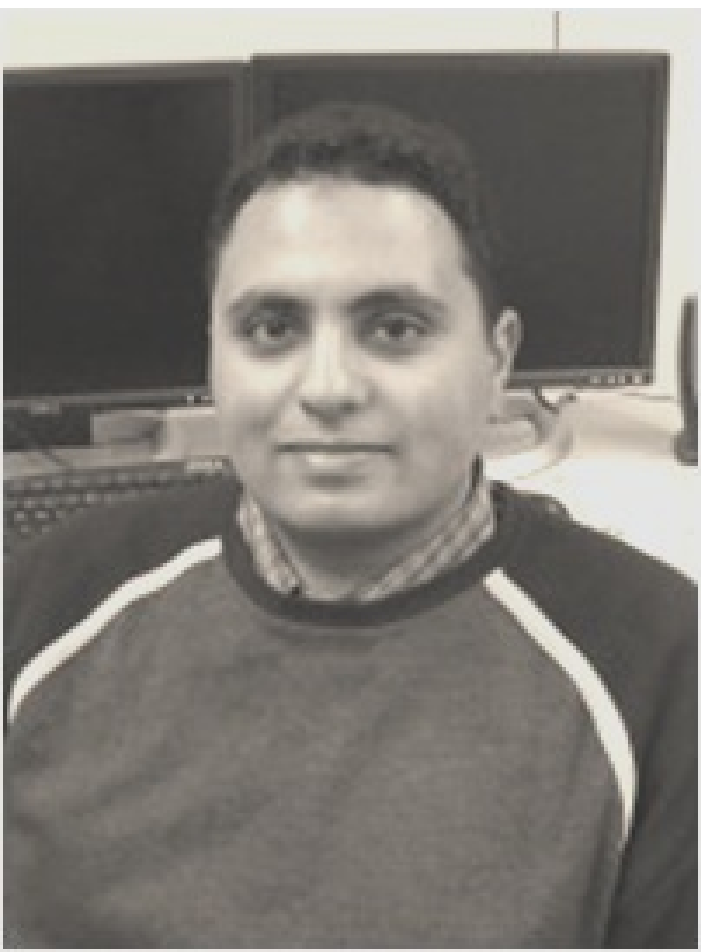}}]{Sameh Sorour} (S '98) received the B.Sc. and M.Sc. degrees in Electrical Engineering from Alexandria University, Egypt, in 2002 and 2006, respectively. He then received his Ph.D from at the Wireless and Internet Research Laboratory (WIRLab), Department of Electrical and Computer Engineering, University of Toronto, Canada. He is currently a Postdoctoral Fellow at University of Toronto. In 2002, he joined the Department of Electrical Engineering, Alexandria
University, where he was a Teaching and Research Assistant for three years and was promoted to Assistant Lecturer in 2006. He is also the chair of local arrangements for IEEE PIMRC 2011. His research interests include opportunistic, random and instantly decodable network coding applications in wireless networks, vehicular and high speed train networks, indoor localization, adaptive resource allocation, OFDMA, and wireless scheduling.
\end{IEEEbiography}

\begin{IEEEbiography}[{\includegraphics[width=1in,height=1.25in,clip,keepaspectratio]{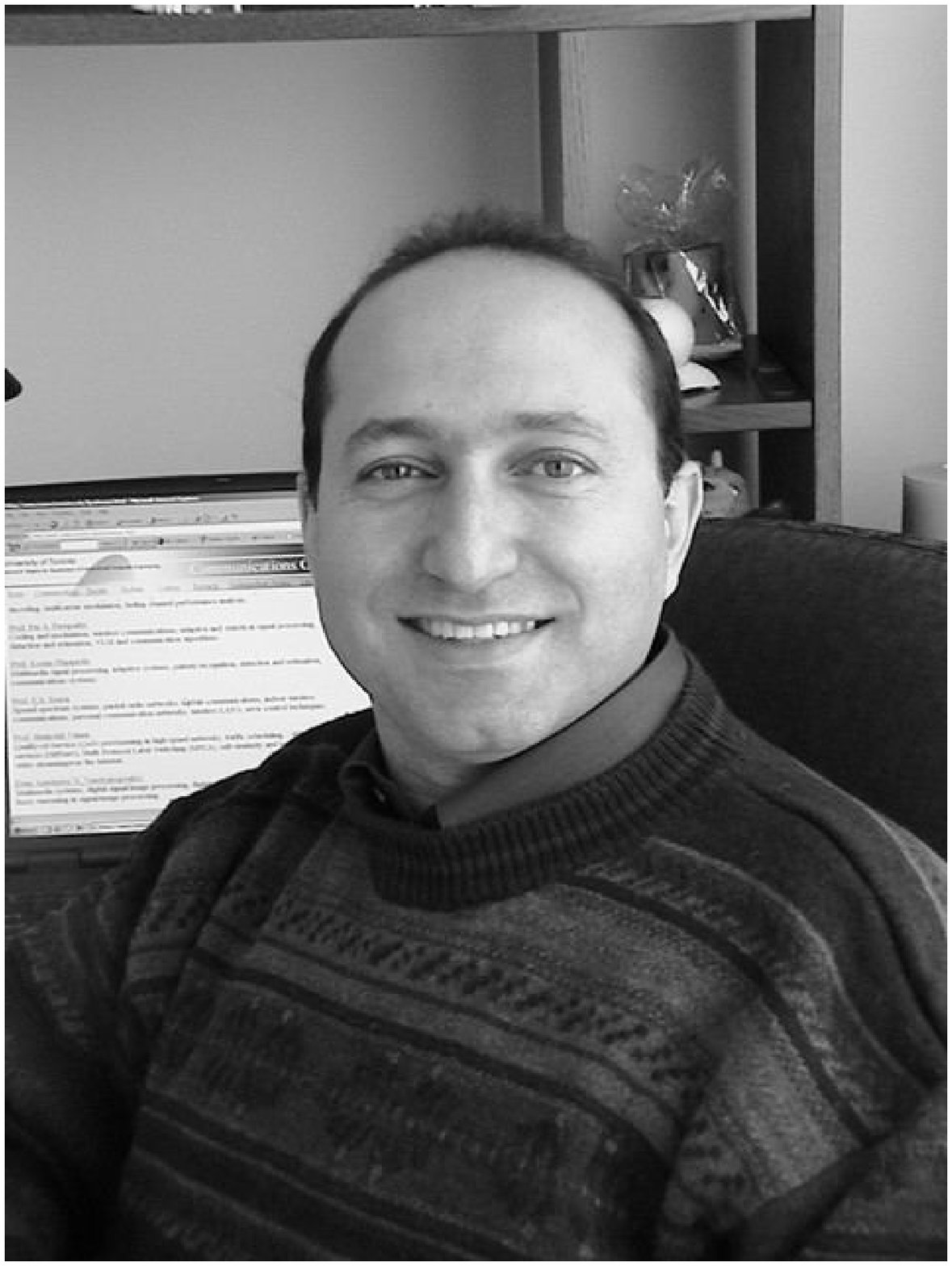}}]{Shahrokh Valaee} (S '88, M '00, SM '02) is the Associate Chair for Undergraduate Studies and the Director of the Wireless and Internet Research Laboratory (WIRLab) in the Edward S. Rogers Sr. Department of Electrical and Computer Engineering, University of Toronto, Canada. Prof. Valaee was the Co-Chair for the Wireless Communications Symposium of IEEE GLOBECOM 2006, a Guest Editor for IEEE Wireless Communications Magazine, a Guest Editor for Wiley Journal on Wireless Communications and Mobile Computing, and a Guest Editor of EURASIP Journal on Advances in Signal Processing. He is currently an Editor of IEEE Transactions on Wireless Communications, an Associate Editor of IEEE Signal Processing Letters, and the TPC-Chair of IEEE PIMRC 2011. His current research interests are in wireless, vehicular and sensor networks, location estimation and cellular networks.
\end{IEEEbiography}


\end{document}